\documentclass[journal]{IEEEtran}

\ifCLASSINFOpdf
  \usepackage[pdftex]{graphicx}

\else

  \usepackage[dvips]{graphicx}

\fi
\usepackage{subfig}
\usepackage{xcolor}
\usepackage{multirow}
\usepackage{booktabs}
\usepackage{cite}
\usepackage{pdfcomment}
\usepackage{amsmath}
\usepackage{amsthm}
\usepackage{amssymb}
\usepackage{caption}
\usepackage{cleveref}
\newtheorem{definition}{Definition}
\newtheorem{proposition}{Proposition}
\newtheorem{lemma}{Lemma}
\newtheorem{remark}{Remark}

\newcommand{\bs}[1]{\mathbf{#1}}
\newcommand{\bb}[1]{\mathbb{#1}}
\newcommand{\mc}[1]{\mathcal{#1}}
\begin{document}
%
% paper title
% Titles are generally capitalized except for words such as a, an, and, as,
% at, but, by, for, in, nor, of, on, or, the, to and up, which are usually
% not capitalized unless they are the first or last word of the title.
% Linebreaks \\ can be used within to get better formatting as desired.
% Do not put math or special symbols in the title.
\title{A Disentangled Representation Learning Framework for Low-altitude Network Coverage Prediction}
%
%
% author names and IEEE memberships
% note positions of commas and nonbreaking spaces ( ~ ) LaTeX will not break
% a structure at a ~ so this keeps an author's name from being broken across
% two lines.
% use \thanks{} to gain access to the first footnote area
% a separate \thanks must be used for each paragraph as LaTeX2e's \thanks
% was not built to handle multiple paragraphs
%

\author{Xiaojie Li,
Zhijie Cai,
Nan Qi,~\IEEEmembership{Senior Member,~IEEE,}
Chao Dong,~\IEEEmembership{Senior Member,~IEEE,}

Guangxu Zhu,~\IEEEmembership{Member,~IEEE,}
Haixia Ma,
%Qingjiang Shi,~\IEEEmembership{Senior Member,~IEEE,}
Qihui Wu,~\IEEEmembership{Fellow,~IEEE,}
and Shi Jin,~\IEEEmembership{Fellow,~IEEE}
% Michael~Shell,~\IEEEmembership{Member,~IEEE,}
%         John~Doe,~\IEEEmembership{Fellow,~OSA,}
%         and~Jane~Doe,~\IEEEmembership{Life~Fellow,~IEEE}% <-this % stops a space
\thanks{The work of Guangxu Zhu was supported by the National Key R\&D Program of China under Grant 2022YFA1003900, Guangdong Major Project of Basic and Applied Basic Research under Grant 2023B0303000001, and Guangdong Young Talent Research Project under Grant 2023TQ07A708. The work of Nan Qi was supported in part by the National Natural Science Foundation of China (No. 62271253, 62471491), in part by the National Aerospace Science Foundation of China under Grant 2023Z021052002, and National Key Laboratory of Wireless Communications Foundation under Grant IFN202411 (Corresponding author: Guangxu Zhu, Nan Qi).}

\thanks{Xiaojie Li is with the School of Information Science and Engineering, Southeast University, Nanjing 210096,
China, also with the College of Physics, Nanjing University of Aeronautics and Astronautics, Nanjing 210016, China, and also with the Shenzhen Research Institute of Big Data, The Chinese University of Hong Kong-Shenzhen, Guangdong 518172, China (e-mail: xiaojieli@nuaa.edu.cn).

Nan Qi, Chao Dong and Qihui Wu are with the Key Laboratory of Dynamic Cognitive System of Electromagnetic
Spectrum Space, Ministry of Industry and Information Technology, Nanjing
University of Aeronautics and Astronautics, Nanjing 210023, China (e-mail: nanqi.commun@gmail.com, dch@nuaa.edu.cn, and wuqihui@nuaa.edu.cn).

Zhijie Cai and Guangxu Zhu are with the Shenzhen Research Institute of Big Data, The Chinese University of Hong Kong-Shenzhen, Guangdong 518172, China (e-mail: zhijiecai@link.cuhk.edu.cn and gxzhu@sribd.cn).

Haixia Ma is with the College of Physics, Nanjing University of Aeronautics and Astronautics, Nanjing 210016, China (e-mail: mahaixia@nuaa.edu.cn).

%Qingjiang Shi is with the School of Software Engineering, Tongji University, Shanghai, China, and with the Shenzhen Research Institute of Big Data, Shenzhen, China (e-mail: shiqj@tongji.edu.cn).

 %is with the Key Laboratory of Dynamic Cognitive System of Electromagnetic Spectrum Space, Nanjing University of Aeronautics and Astronautics, Nanjing 210016, China (e-mail:).

 Shi Jin is with the School of Information Science and Engineering, Southeast University, Nanjing 210096, China (e-mail: jinshi@seu.edu.cn).
}}

% note the % following the last \IEEEmembership and also \thanks - 
% these prevent an unwanted space from occurring between the last author name
% and the end of the author line. i.e., if you had this:
% 
% \author{....lastname \thanks{...} \thanks{...} }
%                     ^------------^------------^----Do not want these spaces!
%
% a space would be appended to the last name and could cause every name on that
% line to be shifted left slightly. This is one of those "LaTeX things". For
% instance, "\textbf{A} \textbf{B}" will typeset as "A B" not "AB". To get
% "AB" then you have to do: "\textbf{A}\textbf{B}"
% \thanks is no different in this regard, so shield the last } of each \thanks
% that ends a line with a % and do not let a space in before the next \thanks.
% Spaces after \IEEEmembership other than the last one are OK (and needed) as
% you are supposed to have spaces between the names. For what it is worth,
% this is a minor point as most people would not even notice if the said evil
% space somehow managed to creep in.

% The paper headers
\markboth{IEEE Transactions on Mobile Computing}%
{Shell \MakeLowercase{\textit{et al.}}: Bare Demo of IEEEtran.cls for IEEE Journals}

% make the title area
\maketitle

% As a general rule, do not put math, special symbols or citations
% in the abstract or keywords.
\begin{abstract}
The expansion of the low-altitude economy has underscored the significance of Low-Altitude Network Coverage (LANC) prediction for designing aerial corridors. While accurate LANC forecasting hinges on the antenna beam patterns of Base Stations (BSs), these patterns are typically proprietary and not readily accessible. Operational parameters of BSs, which inherently contain beam information, offer an opportunity for data-driven low-altitude coverage prediction. However, collecting extensive low-altitude road test data is cost-prohibitive, often yielding only sparse samples per BS. This scarcity results in two primary challenges: imbalanced feature sampling due to limited variability in high-dimensional operational parameters against the backdrop of substantial changes in low-dimensional sampling locations, and diminished generalizability stemming from insufficient data samples. To overcome these obstacles, we introduce a dual strategy comprising expert knowledge-based feature compression and disentangled representation learning. The former reduces feature space complexity by leveraging communications expertise, while the latter enhances model generalizability through the integration of propagation models and distinct subnetworks that capture and aggregate the semantic representations of latent features. 
{Experimental evaluation confirms the efficacy of our framework, yielding a $7\%$ reduction in error compared to the best baseline algorithm. Real-network validations further attest to its reliability, achieving practical prediction accuracy with MAE errors at the $5\ \mathrm{dB}$ level.}

\end{abstract}

% Note that keywords are not normally used for peerreview papers.

\IEEEpeerreviewmaketitle

\section{Introduction}

%\IEEEPARstart{W}{ith} the advent of the 5G Advanced (5GA) and 6th-generation (6G) era, an increasing number of user equipment is being integrated into mobile networks\cite{b1,b2,b3}. To address the dynamic Quality of Service (QoS) demands of these devices, intelligent network planning and optimization to enhance Key Performance Indicators (KPIs) are imperative\cite{NetworkOptimization_2020}. Among these, the modeling and prediction of network coverage quality play a pivotal role in providing crucial evaluative feedback\cite{SRCON_2023}. Consequently, accurately predicting network coverage remains an open issue, consistently drawing research interest.

\IEEEPARstart{T}{he} rapid development of the low-altitude economy, driven by emerging applications such as unmanned aerial vehicles (UAVs), low-altitude logistics, and aerial surveillance, has generated an urgent demand for reliable and accurate Low-altitude Network Coverage (LANC)\cite{Jiang2025survey}. Ensuring seamless connectivity in such scenarios is crucial for enabling efficient network management\cite{rihan2024unified}. As a result, the accurate modeling and prediction of LANC have become a prominent research focus.

Conventionally, existing 4G and 5G network coverage prediction research {has focused on the coverage prediction of terrestrial networks}. For instance, in certain typical scenarios\cite{huang2022artificial}, such as indoor or outdoor, rural or urban settings, empirical models \cite{Abhayawardhana_Wassell_Crosby_Sellars_Brown_2005} or stochastic geometric models\cite{Bian_Wang_Gao_You_Zhang_2021,wang2018survey,winner2007winner, Baum_Hansen_Galdo_Milojevic_Salo_Kyosti_2005} have been employed. However, these approaches lack the capability to describe specific environments in detail \cite{zeng2024tutorial}, thereby limiting accuracy. To address this, ray tracing\cite{zhang2024deterministic,yuan2024efficient,chen2021channel} has been utilized in conjunction with local map modeling to simulate coverage in specific environments. Nevertheless, the challenges of accurately modeling environments and the high computational complexity remain significant hurdles for ray tracing technology\cite{Zeng_Xu_2021,de2021convergent}. Moreover, advanced deep learning methods have also been applied to coverage prediction\cite{zhang2023rme,10682510,minovski2021throughput,li2024deep,al2024machine}. Most of these methods\cite{zhang2023rme,10682510} have achieved high prediction accuracy based on simulated datasets, and some works based on real-world data\cite{minovski2021throughput,li2024deep,al2024machine} have demonstrated the potential of this approach.

%With the burgeoning development of the low-altitude economy\cite{jiang20236g}, researchers have shifted their focus toward the prediction of Low-altitude Network Coverage (LANC)\cite{rihan2024unified}. In contrast to terrestrial networks, low-altitude networks will be planned in airspace with fewer buildings, thereby ensuring enhanced safety. This will lead to the predominance of received signals from direct paths. Consequently, the actual LANC is largely influenced by the multi-beam capabilities of base stations (BSs) and the collaborative effect of multiple BSs\cite{bernabe2024massive}. These characteristics introduce new unique challenges and opportunities for predicting LANC, necessitating novel approaches to address them.
\begin{table}[t]
\caption{Difference between Terrestrial and Low-altitude Networks}
\label{New_table1}
\centering
\begin{tabular}{cccc}
\toprule
                      & Main Path & Beam Impact               & Coverage \\ \midrule
Low-altitude Networks  & Direct Path                & Big & 3D                \\
Terrestrial Networks & Multipath                & Small          & 2D              
\\ \bottomrule
\end{tabular}
\end{table}
With the burgeoning development of the low-altitude economy, researchers have redirected their attention to LANC prediction\cite{rihan2024unified}. {As shown in Table \ref{New_table1}, there are three main differences between terrestrial and low-altitude networks. Unlike} terrestrial networks, low-altitude networks will be planned in airspace with fewer buildings, leading to the predominance of received signals from direct paths. Consequently, the actual LANC is largely influenced by the multi-beam capabilities of base stations (BSs) and the collaborative effect of multiple BSs\cite{bernabe2024massive}. {This shifts the focus from conventional environmental scattering to base station antenna configurations and beamforming patterns. Moreover, many terrestrial models treat the vertical dimension as negligible, making them unsuitable for the three-dimensional nature of low-altitude networks. The traditional terrestrial solutions, which are primarily designed for ground-level environments, are inadequate for addressing these challenges. As a result, there is a clear need for novel approaches specifically tailored to the unique requirements of low-altitude networks.}

In response to this need, many recent studies have attempted to explore the prediction of LANC. However, most of these studies primarily concentrate on 3-dimensional (3D) coverage interpolation for individual regions, i.e., 3D radio map reconstruction. These approaches seek to enhance reconstruction accuracy by incorporating factors such as trajectory optimization\cite{wang2023sparse} and channel shadowing\cite{wang2024sparse}, relying on sparse sampling within the area of interest and lacking the capability for cross-regional generalization. Such cross-regional generalization is critical for practical network planning and deployment, as it enables reliable coverage prediction in previously unseen areas without the need for extensive on-site measurements. Additionally, some efforts have been made towards theoretical modeling of ground-to-air channels within low-altitude networks\cite{lyu2019network,armeniakos2022performance,maeng2024uav}. Nevertheless, a limited understanding of real-world BS models, i.e., the actual beam setting, has constrained their performance. In summary, there is still a notable gap in research that considers the actual multi-beam effects of BSs and possesses the capability for generalizing LANC prediction across various regions.

In reality, if we could access the multi-beam directionality diagrams of BS antennas\cite{zhang2023physics} in cellular networks, it would be feasible to achieve generalization across BSs and even across regions. However, the direct representation of BS antenna multi-beam directionality diagrams is proprietary information of the equipment manufacturers, making it difficult to access. Fortunately, BS configuration parameters implicitly contain beam-related information\cite{he2021intelligent}, including the Active Antenna Unit (AAU) type, number of antenna channels, coverage scenario, and factors influencing the beam direction, such as down tilt and azimuth. By effectively leveraging this information, it is possible to infer the impact of multi-beam gain from BSs.

Given the increasing availability of BS operational data and the high cost of large-scale measurements, leveraging learning-based methods for predicting LANC using sparse low-altitude coverage data becomes feasible. However, during road tests, UAVs primarily capture RSRP from the serving base station within a limited coverage area. Despite traveling several kilometers, the data collected often reflects only a small subset of base stations, typically no more than 30, with each contributing data from a specific region. This results in an imbalanced distribution of data, as high-dimensional base station parameters are sparsely sampled across the entire coverage space, while low-dimensional locations exhibit greater variation. Additionally, the high cost of data collection in early low-altitude economies further limits the compilation of extensive LANC data\cite{gao2023aoi}, restricting the generalization capability of learning models. Ultimately, feature imbalance and data scarcity pose significant challenges to developing models with robust accuracy and generalization.

As shown in Fig. \ref{fig:0}, we have designed a disentangled representation learning framework to overcome the proposed obstacles. On one hand, embedding expert knowledge, such as signal propagation models, can enhance model generalization with limited samples. By leveraging this expertise for feature compression, we obtain a more balanced and compressed feature space, which helps mitigate generalization challenges. On the other hand, disentangled representation learning\cite{wang2024disentangled}, which focuses on learning the latent representations of distinct independent factors, has demonstrated its effectiveness in improving model interpretability, controllability, robustness, and generalization across several fields, including computer vision, natural language processing, and data mining. Inspired by this approach, we categorize the compressed features based on different independent aspects \cite{s23104775} of signal propagation, which enables the dissection of various components affecting signal propagation, thereby enhancing prediction accuracy and generalization.

\begin{figure}
    \centering
    \includegraphics[width=0.8\linewidth]{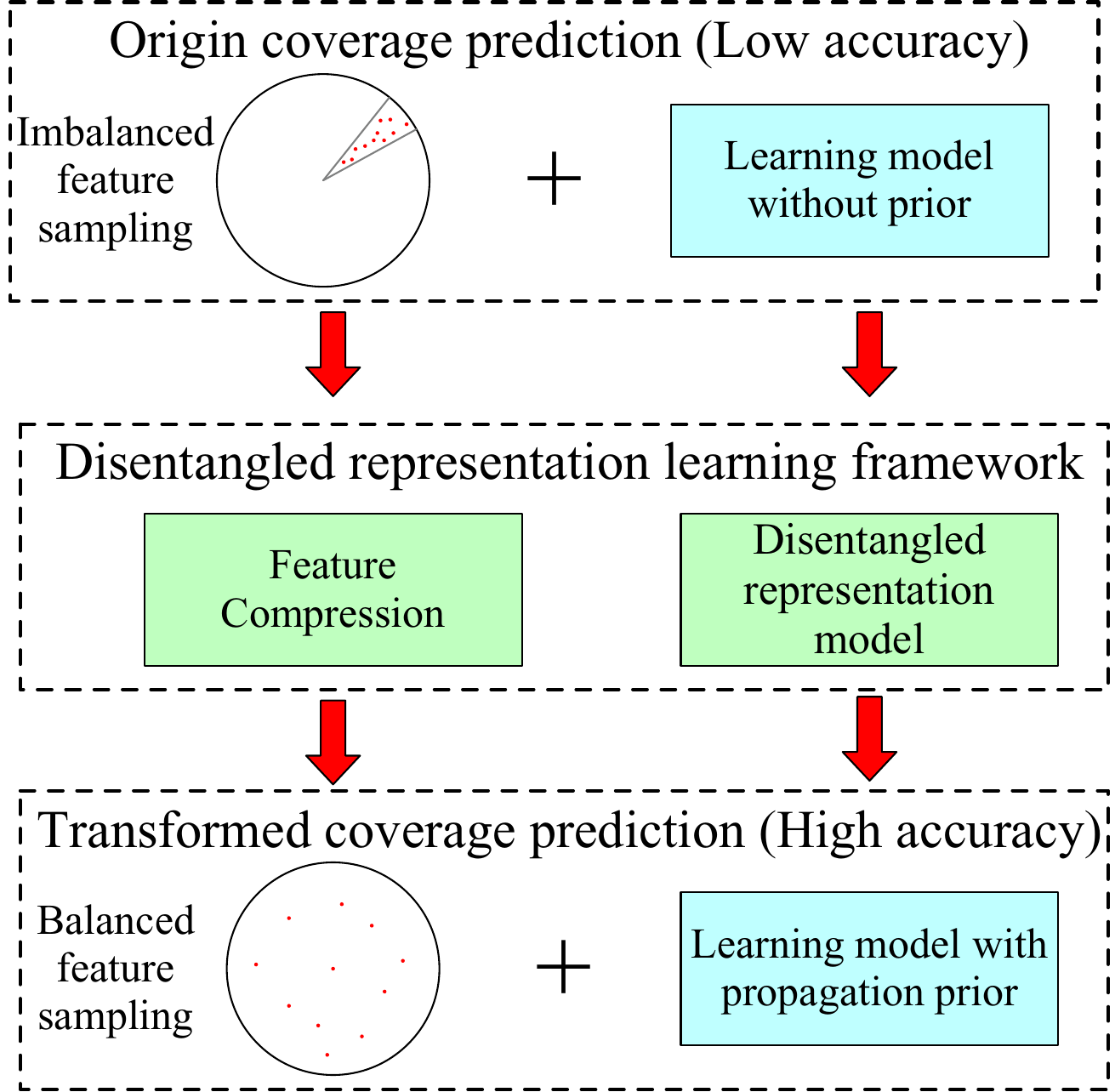}
    \caption{Disentangled representation learning framework.}
    \label{fig:0}
\end{figure}

In summary, we propose a disentangled representation learning framework for the prediction of LANC based on BS operational parameters. Our contributions can be summarized as follows:
\begin{itemize}
\item \textbf{Feature Compression Based on Expert Knowledge}: As part of the proposed framework, we leverage expert knowledge in the field of communications to compress features of BS operational parameters and sparse measurement data in the original feature space. This feature compression effectively addresses the issue of feature imbalance caused by imbalanced sampling across the feature space, providing a solid foundation for accurate and efficient LANC prediction.

\item \textbf{Model-guided Disentangled Representation Network}: Our framework incorporates a disentangled representation network guided by the signal propagation model, which enhances both generalizability and interpretability. The theoretical analysis demonstrates that the proposed network, while using fewer parameters, maintains comparable learning capabilities to traditional deep neural networks (DNNs) in LANC prediction tasks. This reduction in model complexity ensures improved generalization performance, particularly under limited data conditions.

\item \textbf{Experiments with Real-world Network Data}: We validated the proposed framework through real-world low-altitude data collection and network training conducted in Nanchang, Jiangxi Province, China. Ablation studies demonstrate our framework reduces the mean absolute error (MAE) by approximately $20\%$ compared to traditional DNNs. Furthermore, our method outperforms the best baseline algorithm with a $7\%$ reduction in error. In practical LANC prediction of an unknown area, the framework achieves a measured MAE error at the $5\ \mathrm{dB}$ level, demonstrating its reliability and applicability in real-world scenarios.

\end{itemize}
The remainder of this work is organized as follows: 
Section \ref{Section: System Model} introduces the system model and problem formulation. Following this, the proposed disentangled representation learning framework is detailed in Section \ref{Section:Framework}. Section \ref{Section: Theoretical Analysis} provides a theoretical analysis of the proposed framework. Empirical evaluations, including ablation studies, comparisons with existing algorithms on real-world datasets, and further performance validation, are presented in Section \ref{Section: Experiment}. Finally, Section \ref{Section: Conclusion} concludes the paper.

\section{System Model and Problem Formulation} \label{Section: System Model}

\subsection{BS and Its Operational Parameters}

In scenarios such as BS pre-planning or where low-altitude sparse data collection is absent, predicting low-altitude coverage based on operational parameters becomes a prevalent approach. {For such a BS of interest, the operational parameters crucially linked to coverage can be systematically categorised based on their roles in determining coverage. Given the diverse nature of BS parameters, which number in the hundreds across physical and network layers, it is essential to establish clear screening and categorisation criteria for a robust coverage prediction framework.}

{To address this challenge, we employ a systematic parameter selection approach guided by the principle of \textbf{coverage correlation}. Specifically, we retain only those parameters that directly influence the electromagnetic propagation characteristics, beam pattern, or signal strength distribution in the coverage area. This criterion enables us to exclude parameters irrelevant to physical coverage, such as core network latency, backhaul configuration, or higher-layer protocol settings, which do not affect the radio propagation environment. Through this screening process, we distil the extensive parameter space into a manageable set of coverage-critical parameters.}

{Following parameter selection, we establish a logical categorisation framework that organises the selected parameters according to their distinct roles in coverage determination. This categorisation yields four functional groups, each addressing a specific aspect of network coverage:}

\begin{itemize} 
\item \textbf{Absolute Location Attributes:} The parameters of \textit{longitude}, \textit{latitude}, and \textit{antenna height} are pivotal in pinpointing the absolute position of the BSs. These locational attributes provide a foundational understanding of the BS's placement relative to the coverage area.

\item \textbf{Static Characteristics of the Antenna Beam:} The \textit{AAU type}, \textit{number of channels}, \textit{coverage scenario}, and \textit{carrier frequency} are integral in defining the static characteristics of the antenna beam. Specifically, the AAU type dictates the availability of channels at the BS, while the configuration of coverage scenarios allows for the emission of signals aligned with the beam antenna patterns tailored to each scenario at a corresponding carrier frequency. These characteristics collectively delineate the static characteristics of the antenna beam.

\item \textbf{Orientation Characteristics of the Antenna Beam:} The orientation characteristics of the antenna beam are uniquely determined by the \textit{horizontal azimuth}, \textit{beam azimuth}, \textit{mechanical down tilt}, and \textit{digital down tilt}. Adjusting these angles enables precise control over the antenna's orientation, affecting coverage in both horizontal and vertical planes.

\item \textbf{Additive Signal Strength Features:} The \textit{total transmission power} and \textit{bandwidth} are categorized as additive signal strength features. Given a carrier frequency, the total bandwidth collectively establishes the effective utilization bandwidth for the antenna beams emitted by the BS. The antenna beams' transmission power is derived from the total transmission power.
\end{itemize}
For ease of description, the aforementioned categories and the included operational parameters are defined in Table \ref{table1}.

\begin{table}[t]
\centering
\caption{\centering{Category of Operational Parameters and Corresponding Mathematical Expression}}
\label{table1}
\begin{tabular}{cccc}
\toprule
Category           &        Vector       &Parameters&            Scalar                                       \\  \midrule
\multirow{3}{*}{\begin{tabular}[c]{@{}c@{}}Absolute Location\\  Attributes\end{tabular}}                    & \multirow{3}{*}{$\mathbf{x^{(L)}}$} & Longitude            & $x^{(L)}_1$  \\ && Latitude             & $x^{(L)}_2$\\&  & Antenna Height  & $x^{(L)}_3$ \\
\multirow{4}{*}{\begin{tabular}[c]{@{}c@{}}Static Characteristics\\  of the Antenna Beam\end{tabular}}      & \multirow{3}{*}{$\mathbf{x^{(S)}}$} &                     AAU Type & $x^{(S)}_1$                     \\ &    &     Num of Channels & $x^{(S)}_2$                     \\ &  &     Coverage Scenario& $x^{(S)}_3$          \\ &    &     Carrier Frequency & $x^{(S)}_4$           \\
\multirow{4}{*}{\begin{tabular}[c]{@{}c@{}}Orientation\\ Characteristics of\\ the Antenna Beam\end{tabular}} & \multirow{4}{*}{$\mathbf{x^{(O)}}$} &     Horizontal Azimuth    & $x^{(O)}_1$                     \\  &  &     Beam Azimuth & $x^{(O)}_2$                     \\ &   &                      Mechanical Down Tilt& $x^{(O)}_3$                     \\& &               Digital Down Tilt & $x^{(O)}_4$                     \\
\multirow{3}{*}{\begin{tabular}[c]{@{}c@{}}Additive Signal\\  Strength Features\end{tabular}}      & \multirow{2}{*}{$\mathbf{x^{(A)}}$} &                     Total Transmission Power & $x^{(A)}_1$                                          \\ &  &     Bandwidth& $x^{(A)}_2$  \\ \bottomrule                   
\end{tabular}
\end{table}

\subsection{Low-altitude Coverage for Area of Interest}

\begin{figure}
    \centering
    \includegraphics[width=0.9\linewidth]{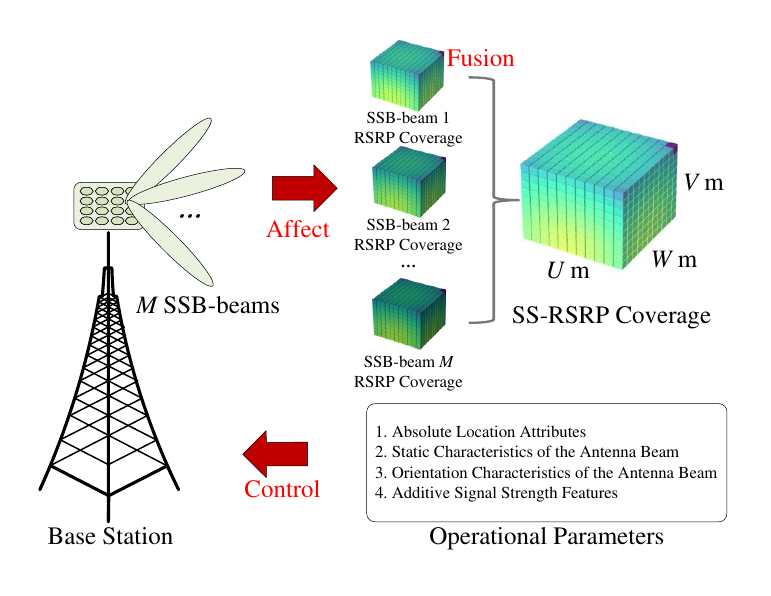}
    \caption{The process of low-altitude coverage prediction based on operational parameters.}
    \label{fig:1}
\end{figure}

To further explore the LANC provided by the BS, consider a nearby area with dimensions $U \times W \times V$ m\(^3\), corresponding to its longitudinal, latitudinal, and altitudinal extents. This area can be subdivided into $N$ cubes, where $N=N_U \times N_W \times N_V$ and $N_U$, $N_W$, and $N_V$ represent the counts of discretization along the longitude, latitude, and altitude axes, respectively. The set of cubes can thus be defined as $\mathbf{R}=\{1,\cdots,r_n,\cdots,r_N\}$. For cube $r_n$, the absolute position vector can be defined as $\mathbf{x^{(R)}_n}=[x^{(R)}_{n,1},x^{(R)}_{n,2},x^{(R)}_{n,3}]^T$, where $x^{(R)}_{n,1}$, $x^{(R)}_{n,2}$ and $x^{(R)}_{n,3}$ are longitude, latitude, and altitude, respectively.
Furthermore, the absolute position array for all cubes is defined as $\mathbf{X^{(R)}}=[\mathbf{x^{(R)}_1},\cdots,\mathbf{x^{(R)}_n},\cdots,\mathbf{x^{(R)}_N}] \in \mathbb{R}^{3\times N}$. {Next, we will further elaborate on signal strength measurements and introduce cube-level coverage.}

{In practical implementation, signal strength measurements are obtained through low-altitude drive testing procedures using dedicated measurement terminals provided by network equipment manufacturers. These terminals systematically collect geographical coordinates, monitored frequency bands, Physical Cell Identity (PCI) of serving base stations, and comprehensive beam-specific signal measurements.}

{To establish the theoretical foundation for low-altitude coverage analysis, it is essential to understand the beam structure in 5G NR networks. The current 5G protocol defines three primary downlink beam categories: (i) SSB beams (static beams utilized for initial terminal access and synchronization procedures), (ii) CSI-RS beams (dynamic beams employed for channel state information estimation), and (iii) PDSCH beams (dynamic beams dedicated to data transmission). For low-altitude network coverage prediction, terminal connectivity to target base stations is predominantly determined by SSB beam quality, as these beams establish the fundamental synchronisation and access framework required for initial network attachment.}

{SSB-RSRP (Synchronisation Signal Block - Reference Signal Received Power) quantifies the signal strength of individual SSB beams transmitted by a base station. Within each base station's operational framework, multiple SSB beams are transmitted using time-division multiplexing across different spatial directions to provide comprehensive coverage. SS-RSRP represents the maximum SSB-RSRP value among all transmitted SSB beams from a specific base station, effectively indicating the strongest achievable signal strength for terminal connectivity at any given spatial location.}

Determined by the AAU type, the BS will emit $M$ beams. For any given cube $n$, $M$ SSB-RSRP measurements can be obtained, along with a maximum SSB-RSRP, denoted as SS-RSRP. {These measurements are mathematically represented as $\mathbf{p_{n}}=[p_{n,0},p_{n,1},\cdots,p_{n,m},\cdots,p_{n,M}]^T \in \mathbb{R}^{(M+1)\times 1}$, where $p_{n,0}$ represents the SS-RSRP for the $n$-th cube, and $p_{n,m}$ ($m=1,2,\ldots,M$) denotes the $m$-th individual SSB-RSRP measurement for the $n$-th cube. Consequently, the comprehensive coverage characterisation for all cubes within the target area is expressed as $\mathbf{P}=[\mathbf{p_{1}},\cdots,\mathbf{p_{n}},\cdots,\mathbf{p_{N}}]\in \mathbb{R}^{(M+1)\times N}$.}
\subsection{Problem Formulation}
As shown in Fig. \ref{fig:1}, the objective of the study is to identify a learning machine, denoted as $g$, capable of predicting multi-beam RSRP coverage at the cube level, based on given operational parameters. For this task, the absolute position features of the cubes and the BS should be transformed into relative features. Specifically, for any cube $r_n$, the relative position feature is defined as
\begin{equation}
   \Delta \mathbf{x^{(R)}_n} = \mathbf{x^{(R)}_n} - \mathbf{x^{(L)}}.
\end{equation}
Thus, at cube $r_n$, the prediction of multi-beam RSRP $\mathbf{p_n}$ can be denoted as\footnote{The prediction model in Eq. (2) does not account for the shadowing effects caused by the surrounding environment of the BS deployment. This omission is considered reasonable for low-altitude coverage predictions at heights of 50 meters or above. At such altitudes, the coverage performance is primarily influenced by beamforming and the attenuation of line-of-sight (LOS) links. The potential gains resulting from environmental factors can be discussed in future work.}
\begin{equation}
   \mathbf{\hat{p}_n}=g(\mathbf{x^{(S)}},\mathbf{x^{(O)}},\mathbf{x^{(A)}},\Delta \mathbf{x^{(R)}_n}).
\label{Prediction Task}
\end{equation}
For the whole area, the prediction of $\mathbf{P}$ can be denoted as $\mathbf{\hat{P}}=[\mathbf{\hat{p}_{1}},\cdots,\mathbf{\hat{p}_{n}},\cdots,\mathbf{\hat{p}_{N}}]\in \mathbb{R}^{(M+1)\times N}$.
Finally, the problem of coverage prediction can be  formulated as
\begin{equation}
    \min_{g} L(\mathbf{P},\mathbf{\hat{P}}),
    \label{Eq Problem}
\end{equation}
where $L$ denotes a loss function that quantifies the discrepancy between $\mathbf{\hat{P}}$ and $\mathbf{P}$. The primary aim of the learning machine $g$ is to minimize this loss, thereby enabling precise prediction of the RSRP coverage performance at each point within the area of interest, based on the operational parameters. 

To tackle Problem (\ref{Eq Problem}), an intuitive approach involves leveraging sparse measurement samples from additional BSs and employing conventional supervised learning frameworks to develop regression models, such as Random Forest. However, this approach faces significant challenges due to imbalanced feature sampling. Specifically, with just a 2° discretization for variables like the horizontal azimuth, beam azimuth, mechanical down tilt, and digital down tilt, the parameter search space expands to the magnitude of 
\(10^8\) (i.e., \(360/2 \times 360/2 \times 180/2 \times 180/2\) permutations). UAV-based low-altitude road measurements over a 3km×3km area, however, engage only a few dozen BSs (around 50 combinations of operational parameters). This poses significant challenges for the generalization capability of conventional approaches. To address these challenges, we propose a disentangled representation learning framework that enhances the prediction and generalization performance for Eq. (\ref{Eq Problem}) through feature compression based on expert knowledge and the design of a disentangled representation neural network.

\section{Disentangled Representation Learning Framework}\label{Section:Framework}
\subsection{Expert-knowledge-based Feature Compressing}
In Section \ref{Section: System Model}, we preliminarily categorized the sample features based on expert knowledge into \(\mathbf{x^{(S)}}, \mathbf{x^{(O)}}, \mathbf{x^{(A)}},\) and \(\Delta \mathbf{x^{(R)}_n}\). It is important to note that when considering \(\{\mathbf{x^{(S)}}, \mathbf{x^{(O)}}, \mathbf{x^{(A)}}, \Delta \mathbf{x^{(R)}_n}\}\) as a complete set of sample features, the variation among different samples is primarily concentrated in the three-dimensional \(\Delta \mathbf{x^{(R)}_n}\), while the ten-dimensional variation of \(\{\mathbf{x^{(S)}}, \mathbf{x^{(O)}}, \mathbf{x^{(A)}}\}\) is relatively minor. This directly leads to sample imbalance. In fact, only \(\mathbf{x^{(S)}}\) is solely related to the type of BS and exhibits less variation. \(\mathbf{x^{(O)}}\) can be combined with \(\Delta \mathbf{x^{(R)}_n}\) to calculate the relative angle of the target cube with respect to the antenna panel's facing direction. Specifically, the horizontal and vertical angles of cube \(r_n\) relative to the BS's true north direction are first calculated using \(\Delta \mathbf{x^{(R)}_n}\), i.e.,
\begin{equation}
     \theta_n^{(H)}=\arctan2(\Delta x^{(R)}_{n,1},\Delta x^{(R)}_{n,2}),\label{arctan2}
\end{equation} 
\begin{equation}
    \theta_n^{(V)}=\arctan(\Delta x^{(R)}_{n,3}/\sqrt{{\Delta x^{(R)}_{n,1}}^2+{\Delta x^{(R)}_{n,2}}^2}),
\end{equation}
where $ \theta_n^{(H)}$ and $\theta_n^{(V)}$ are the horizontal and vertical angles of cube \(r_n\) relative to the BS's true north direction, respectively. The elements $\Delta x^{(R)}_{n,1}$, $\Delta x^{(R)}_{n,2}$, and $\Delta x^{(R)}_{n,3}$ represent the first, second, and third components of $\Delta \mathbf{x^{(R)}_n}$, respectively. These components signify the differences in longitude, latitude, and altitude dimensions between cube $r_n$ and the considered BS. Eq. (\ref{arctan2}) returns the horizontal relative angle of \(r_n\) with respect to the true north direction of the considered BS, with a range of \([- \pi, \pi]\) (in radians).

It is noted that the angles of the elements within \(\mathbf{x^{(O)}}\) are relative to the BS's true north direction. Consequently, the relative angle of cube \(n\) with respect to the antenna panel's normal direction (denoted as $\Delta \theta_n^{(H)}$ for horizontal and $\Delta \theta_n^{(V)}$ for vertical) can be computed as follows:

\begin{equation}
     \Delta \theta_n^{(H)}=\theta_n^{(H)}-x^{(O)}_{1}-x^{(O)}_{2},
     \label{Eq_7}
\end{equation}
\begin{equation}
    \Delta \theta_n^{(V)}=\theta_n^{(V)}-x^{(O)}_{3}-x^{(O)}_{4}. \footnote{Theoretically, variations in the digital down tilt of a BS are not entirely equivalent to changes in the mechanical down tilt. For practical applications, it is acceptable to consider the total down tilt of a BS as the sum of the digital down tilt and the mechanical down tilt.}
\end{equation}

Since \(\Delta \theta_n^{(H)}\) and \(\Delta \theta_n^{(V)}\) already characterize the relative angular relationship between cube \(n\) and the antenna panel, the angular information contained in \(\Delta \mathbf{x^{(R)}_n}\) is redundant. Given the orthogonality of the polar coordinate system, it is only necessary to extract the distance feature from \(\Delta \mathbf{x^{(R)}_n}\), i.e.,
\begin{equation}
    D_n=\sqrt{{\Delta x^{(R)}_{n,1}}^2+{\Delta x^{(R)}_{n,2}}^2+{\Delta x^{(R)}_{n,3}}^2}.
    \label{Eq_8}
\end{equation}

Now, we have transform the imbalanced variation of \(\Delta \mathbf{x^{(R)}_n}\) and $\mathbf{x^{(O)}}$ into that of $\Delta \theta_n^{(H)}$, $\Delta \theta_n^{(V)}$ and $D_n$ along with $n$. Next, we will decouple the influence of \(\mathbf{x^{(A)}}\) based on the signal propagation model. Generally speaking, the received SSB-RSRP, \(p_{n,m}\), at location \(r_n\) for the \(m\)th beam can be described as
\begin{equation}
p_{n,m} = p^{(T)} + G^{(T)}_{n,m} + PL_{n,m} + G^{(R)}.    \label{Eq.9}
\end{equation}
Here, \(p^{(T)}\) represents the transmission power of the SSB signal, \(G^{(T)}_{n,m}\) denotes the directional gain of the BS for the \(m\)th beam at grid \(r_n\), and \(G^{(R)}\) is the gain of the receiver antenna. $PL_{n,m}$ is the path fading between cube $r_n$ and corresponding beam $m$.\footnote{In the context of  LANC, different beams may follow distinct paths, leading to varied path fading.} Given that subsequent experiments are conducted using the same single-antenna UAV, without loss of generality, we simplify the receiver antenna to an omnidirectional antenna.

Theoretically, \(p^{(T)}\) can be obtained from \(x^{(A)}_1-10\log_{10}x^{(A)}_2\). However, due to the actual SSB signal bandwidth utilization not being 100\% and its susceptibility to the BS's carrier frequency and bandwidth, represented by \(\sigma\)\footnote{In commercial 5G networks, the bandwidth utilisation parameter \(\sigma\) varies according to channel bandwidth and subcarrier spacing (SCS) configurations as specified in 3GPP standards. Based on practical deployments, \(\sigma\) typically ranges from approximately 0.79 to 0.98, with higher values achieved in wider bandwidth allocations. For instance, with 15 kHz SCS, \(\sigma\) ranges from 0.90 (5 MHz) to 0.972 (50 MHz), while 30 kHz SCS achieves values from 0.792 (5 MHz) to 0.983 (100 MHz). In our experimental validation, we consider network configurations where \(\sigma\) ranges from 0.954 to 0.983, representing high-efficiency operating conditions typical of contemporary 5G deployments. This practical grounding ensures our coverage prediction model accurately reflects real-world system performance.}, the actual \(p^{(T)}\) is given by
\begin{equation}
   p^{(T)}=x^{(A)}_1-10\log_{10}x^{(A)}_2-10\log_{10}(\sigma).
\label{Eq_10}
\end{equation}

Thus, the impact of $\mathbf{x^{(A)}}$ has been decoupled. Let $y_{n,m}$ denote $p_{n,m}-p^{(T)}$. Then set $\mathbf{y_{n}}=[y_{n,0},y_{n,1},\cdots,y_{n,m},\cdots,y_{n,M}]^T \in \mathbb{R}^{M\times 1}$. The prediction task (\ref{Prediction Task}) can be transformed into
\begin{equation}
   \mathbf{\hat{y}_n}=g(\mathbf{x^{(S)}},\Delta \theta_n^{(H)}, \Delta \theta_n^{(V)}, D_n).
\label{transformed Prediction task}
\end{equation}

By incorporating expert knowledge, the transformation of the prediction task has addressed the issue of imbalanced feature variation. Simultaneously, it significantly reduces the dimensionality of the input variables, thereby ensuring an enhancement in coverage prediction performance. {The specification of practical \(\sigma\) ranges further ensures that our model parameters reflect realistic system operating conditions, enhancing the reliability of coverage predictions for deployment planning.} Next, we will introduce the corresponding disentangled representation network design.
\subsection{Disentangled-network-based Feature Representation}
\begin{figure}
    \centering
    \includegraphics[width=1\linewidth]{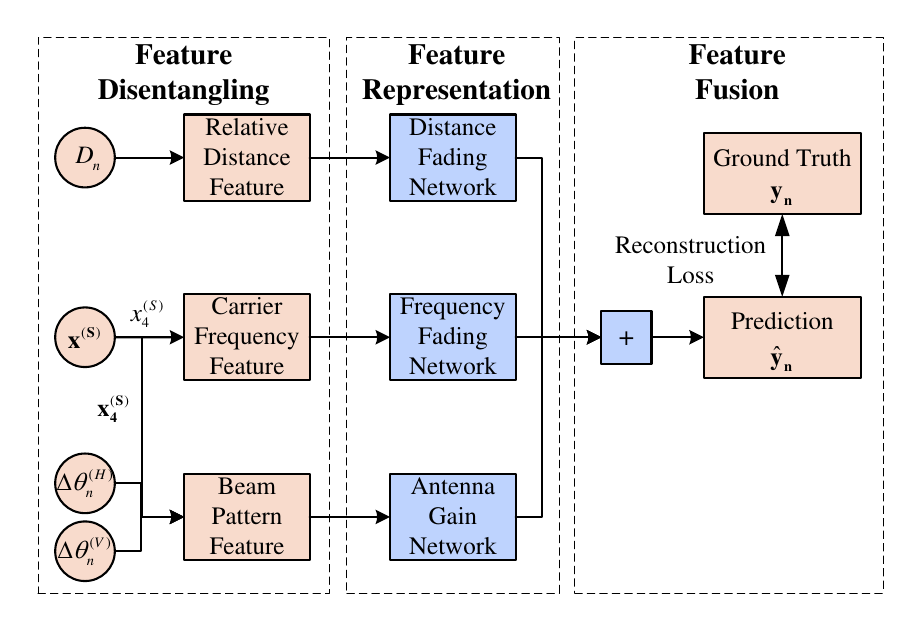}
    \caption{The design of our disentangled representation network, where distance fading, frequency fading, and antenna gain networks can be a simple MLP for feature representation, respectively.}
    \label{Fig Neural Network}
\end{figure}
\begin{figure}
    \centering
    \includegraphics[width=0.45\textwidth]{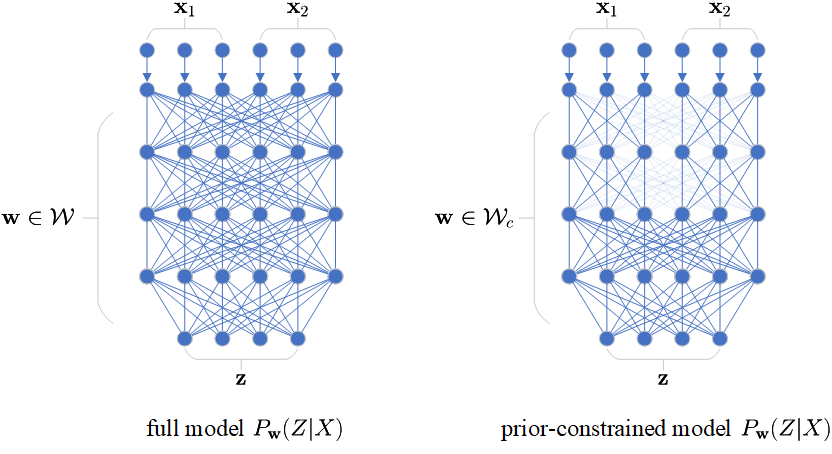}
    \caption{Full model and prior constrained model.}
    \label{fig:3}
\end{figure}
To predict LANC for a target BS, a neural network is initially trained using road test data and the sparse collections from other available BSs. This network is designed to execute the prediction task from input to output, as specified in Eq. (\ref{transformed Prediction task}). Firstly, we will introduce the core motivation behind our network design. According to Eq. (\ref{Eq.9}), the prediction target \(\mathbf{y_n}\) is primarily composed of a linear addition of two parts: \(PL_{n,m}\) and \(G^{(T)}_{n,m}\) (where \(G^{(R)}\) is considered as an omnidirectional antenna and can be omitted for simplicity). Drawing on literature \cite{9523765}, \(PL_{n,m}\) (considering the characteristics of low-altitude networks, shadowing is not taken into account here) can be further expressed as 
\begin{equation}
  PL_{n,m} = \alpha \log_{10} D_n + \beta \log_{10} x^{(S)}_4 .
\end{equation}
 However, the assumptions of this $\log$ model form may limit data fitting performance. Therefore, we adopt a relaxed expression for \(\mathbf{y_n}\), that is: 
\begin{equation}
    \mathbf{y_n} = f_1(\Delta \theta_n^{(H)}, \Delta \theta_n^{(V)}, \mathbf{x_4^{(S)}}) + f_2(D_n) + f_3(x^{(S)}_4).\label{Eq_13}
\end{equation}
Here, $\mathbf{x_4^{(S)}}=[x_1^{(S)},x_2^{(S)},x_3^{(S)}]^T$. \(f_1(\cdot)\), \(f_2(\cdot)\), and \(f_3(\cdot)\) are obtained through fitting with neural networks.

Overall, as illustrated in Fig. \ref{Fig Neural Network}, the design of the neural network comprises three main components: 1) Feature Disentangle, 2) Feature Representation, and 3) Feature Fusion. 
In the following, we will detail each of these components.

\subsubsection{Feature Disentangle} According to Eq. (\ref{Eq_13}), the inputs for the entire network are set to include \(\Delta \theta_n^{(H)}\), \(\Delta \theta_n^{(V)}\), \(\mathbf{x_4^{(S)}}\), \(D_n\), and \(x^{(S)}_4\), where \(\{\Delta \theta_n^{(H)}, \Delta \theta_n^{(V)}, \mathbf{x_4^{(S)}}\}\), \(D_n\), and \(x^{(S)}_4\) are fed into different neural networks for semantic extraction. 
% It is noted in Fig. \ref{Fig Neural Network} that \(x^{(S)}_1\), representing the AAU type, is not included as an input of beam pattern feature. This exclusion is due to the fact that when predicting coverage for a target BS, unless the AAU type of the target station has appeared in the training set, the AAU type of the target station would be treated as a new category. This new category, essentially used for differentiation, does not facilitate the extraction of correlations between stations with different AAU types, thereby leading to a decrease in generalization performance. Hence, \(x^{(S)}_1\) is not utilized for prediction in this context.
\subsubsection{Feature Representation} Here, the relative distance, carrier frequency, and beam pattern features are respectively fed into the distance fading, frequency fading, and antenna gain networks to learn semantic representations. All networks utilize Multi-Layer Perceptron (MLP) architectures. 

Specifically, for the distance fading and frequency fading networks, due to the assumption made in Eq. (9) that, given $r_n$, signals from different beams exhibit different path fading, both networks output a semantic vector of dimension $M+1$. For the antenna gain network, given a BS, its $x_2^{(S)}$ and $x_3^{(S)}$ are fixed, but the sparse sampling positions affect $\Delta \theta_n^{(H)}$ and $\Delta \theta_n^{(V)}$. Therefore, the antenna gain network can learn the semantics of gain across different relative angles in the beam antenna pattern for a given station. Moreover, as the BS changes, $x_2^{(S)}$ and $x_3^{(S)}$ will also change, thereby facilitating the antenna gain network's learning of gain variations across different stations. Given that the antenna patterns vary across beams, the antenna gain network outputs a vector of dimension $M+1$.
\subsubsection{Feature Fusion}
According to Eq. (13), the outputs from different networks can be summed to obtain the predicted $\mathbf{\hat{y}_n}$, which is then compared with the actual sparse measurements $\mathbf{y_n}$ to compute the reconstruction loss. Without loss of generality, we employ MSELoss as the reconstruction loss, defined as 
\begin{equation}
    L(\mathbf{y_n}, \mathbf{\hat{y}_n})  =\frac{1}{M+1} || \mathbf{y_n} - \mathbf{\hat{y}_n})||_2^2,
\end{equation}
where $||\cdot||_2$ denotes the L2 norm of a vector. It is noteworthy that the outputs of the three networks are added directly. Consequently, the partial derivative of the loss with respect to each network's output is 1. This characteristic encourages each network to continually learn and adjust its weights to minimize the reconstruction loss.

\section{Theoretical Analysis of Disentangled Representation Learning Framework}\label{Section: Theoretical Analysis}
In this section, the proposed disentangled representation learning framework will be explained. Initially, evidence for the designed feature disentangling and representation processes will be presented.
\subsection{An Explanation on the Efficacy of Feature Disentangling and Representation}\label{Section IV-A}

It is noted that in Fig. \ref{Fig Neural Network}, $\{\Delta \theta_n^{(H)}, \Delta \theta_n^{(V)}, \mathbf{x_4^{(S)}}\}$, $\{D_n\}$, and $\{x^{(S)}_4\}$ are independent features and we use a single representation for them respectively. This setting can result in lower overfitting risk. We will next give a more general proof.
% \textbf{Notations.} Without further specifications, we use lowercase letters to denote scalars and boldface letters and symbols to denote column vectors. We use uppercase letters to denote random variables.

\subsubsection{Definitions}

\begin{definition}[Decomposable data]
    Consider we are given $K$ sets of signal vectors as a set of realizations of a corresponding independent random variable, e.g., $\{\bs{x}^k_i\}_{i \in [N]} \sim X^k, k \in [K]$. A dataset is formed by concatenating signal vectors from different sets, e.g., $\bs{x}_i = [\bs{x}_i^{1T}, \dots, \bs{x}_i^{KT}]^T$. 
\end{definition}

\begin{definition}[Full model]
    A full model $P_{\bs{w}}(Z|X)$ parameterized by $\bs{w}$ is used to code each $\bs{x}_i$ into some feature $\bs{z}_i$, where $\bs{w} \in \mc{W} \subset \bb{R}^d$, where $\mc{W}$ is the biggest subset allowed by the model structure, $d$ is the number of parameters contained in the model.
\end{definition}

\begin{definition}[Constrained model]
    A constrained model is a full model but with $\bs{w} \in \mc{W}_c \subseteq \mc{W}$, where $\mc{W}_c$ is a constrained subset allowed by the model structure.
\end{definition}

\begin{remark}
    Without abuse of notation, we denote by $\bb{R}$ the set of real numbers. For example, for a common MLP, we can regard that $\mc{W} = \bb{R}^d$. Since modern deep neural network models are large in the number of parameters, the model is often over-parameterized and biased to the seen training data, resulting in a phenomenon known as overfitting. A method to mitigate the effect is to narrow the search space based on some a priori knowledge. For example, by knowing that some of the parameters in the model are redundant, \cite{Paper_dropout} proposed DropOut to stochastically choose some dimensions of $\bs{w}$ and constrain the domain of that dimension to $\{0\}$.
\end{remark}\label{remark_1}

\begin{definition}[Learning capability]
Considering the generalization requirements of LANC prediction, we adopt the generalization error\cite{wang2021analyzing} as a metric to characterize the model's learning capability. Given a model parameter $\bs{w}$, a set of training data $\{\bs{x}_i\}$, and a loss function $\ell(\bs{x}, \bs{w})$, the learning capability is defined as \begin{align}
        \text{gen}(\bs{w}, \{\bs{x}_i\}) := \left|\bb{E}_{\bar{\bs{x}} \sim X}[\ell(\bar{\bs{x}}, \bs{w})] - \bb{E}_{\bs{x} \in \{\bs{x}_i\}}\left[\ell(\bs{x}, \bs{w})\right]\right|.
    \end{align}
\end{definition}

\begin{remark}
The learning capability characterizes the gap between the model's performance on training samples and that on test samples, unseen by the model but assumed to follow the same distribution. A model demonstrating strong generalization will exhibit reduced learning capability disparity.
\end{remark}
\subsubsection{Decomposing model for comparable learning capability}

We try to establish that our proposed prior-constrained model structure can maintain a comparable learning capability to traditional DNNs. From an angle of deep learning practitioners, knowing how $\bs{x}$ is 
independently composed of $[\bs{x}^{1T}, \dots, \bs{x}^{KT}]$, we propose to design dedicated \emph{model heads} for each of the components as shown in Fig. \ref{fig:3}. Then, to make use of the full data, we aggregate the output of the heads as the final feature by concatenation. This procedure is equivalent to a constrained model. To see this, take a toy example of the linear layer in a full model, transforming a vector of $d$ dimensions into one of $d_1$ dimensions. The weight of this layer can be viewed as a matrix $\bs{w} \in \bb{R}^{d \times d_1}$. To gain a corresponding constrained model, some dimension of $\bs{w}$ can be constrained to $\{0\}$.

\begin{lemma}[Learning capability bounded by mutual information \cite{xu2017information}]
    When the loss function $\ell$ is subgaussian, we have
    \begin{align}
        \text{gen}(\bs{w},X) \leq \sqrt{2 \sigma^2 I(\bs{w}; X)}
    \end{align}
    over the randomness of the stochastic learning algorithm, where $X$ denotes the training set as a random variable.
\end{lemma}

\begin{remark}
    Subgaussianity is an easy qualification since a bounded function is always subgaussian.
\end{remark}

\begin{definition}[Equivalent indicator]
    $C$ is a binary random variable, takes value $1$ if a $\bs{w} \in \mc{W}_c$ else $0$.
\end{definition}

\begin{proposition}[Comparable learning capability of constrained model]\label{prop1}
    The constrained model exhibits an upper bound on learning capability that is not greater than its corresponding full model, reflecting a desirable reduction conducive to improved generalization.
\end{proposition}

\begin{proof}
    
We form a Markov chain $X \leftrightarrow \bs{w} \leftrightarrow C$. By Kolmogorov's identity, we have \begin{align}
    I(X, C; \bs{w}) &= I(X; \bs{w}) + \underbrace{I(X; C | \bs{w})}_{=0, \text{Markovian}} \\
    &= I(X; C) + I(X; \bs{w} | C).
\end{align}

Since mutual information $I$ is nonnegative, we will always have $I(X; \bs{w} | C) \leq I(X; \bs{w})$, i.e., the learning capability bound of the constrained model isn't more than that of the full model.
\end{proof}

\begin{proposition}[Lower overfitting risk of the constrained model]\label{prop2}
The constrained model demonstrates a lower risk of overfitting compared to its corresponding full model.
\end{proposition}

\begin{proof}
Based on Proposition \ref{prop1}, the constrained model achieves comparable learning capability. Moreover, the constrained model is characterized by a reduced number of learnable parameters, which inherently limits its capacity to overfit noise or spurious patterns in the training data. This reduction in model complexity effectively lowers the risk of overfitting.
\end{proof}

\begin{table*}[t]
\centering
\caption{\centering{Simple Description of Dataset}}
\label{tableII}
\begin{tabular}{ccccccccc}
\toprule
Area Index & Longitude & Latitude & Length (m) & Width (m) & Height (m)  & Sample Number & BS Number & Size                    \\ \midrule
1          & 115.69    & 28.61    & 4006       & 4273      & 150,300,500 & 43919               & 33                      & \multirow{3}{*}{9.31MB} \\
2          & 115.73    & 28.71    & 1501       & 1262      & 150,300,500 & 25322               & 50                      &                         \\
3          & 115.79    & 28.68    & 1862       & 1134      & 150,300,500 & 24598               & 59                      &                         \\ \bottomrule
\end{tabular}
\end{table*}
\begin{table}[t]
\centering
\caption{\centering{Benchmarks for Ablation Study}}
\begin{tabular}{cccc}
\toprule
Benchmark & Backbone & \begin{tabular}[c]{@{}c@{}}Compressed\\ Feature\end{tabular} & \begin{tabular}[c]{@{}c@{}}Disentangled\\ Network\end{tabular} \\ \midrule
1         & Proposed & $\checkmark$                                                   & $\checkmark$                                                   \\
2         & MLP      & $\checkmark$                                                   &  \textbackslash                                \\
3         & MLP      & \textbackslash                                 & \textbackslash                                 \\ \bottomrule
\end{tabular}
\label{tableIII}
\end{table}
\subsection{An Explanation on the Efficacy of Feature Fusion}
Based on the explanation in subsection \ref{Section IV-A}, we understand that under deterministic prior knowledge, it is possible to artificially set certain parts of the MLP to zero, thereby ensuring a lower overfitting risk. According to Eq. (\ref{Eq_13}), we know that the semantics represented by three separate parts can be summed to obtain the predicted RSRP. To embed this piece of knowledge into the model, we only need to change the feature fusion part in Fig. \ref{fig:3} from an MLP to an addition operation (setting specific parameters in the MLP to zero can achieve this addition). Therefore, by artificially setting parameters to zero once again, we modify the model following Eq. (\ref{Eq_13}), maintaining the comparable learning capability.

\section{Experimental Result}\label{Section: Experiment}

\subsection{Dataset}
To evaluate the performance of the proposed method, initial low-altitude road tests were conducted across three distinct areas within Nanchang City, Jiangxi Province. For these tests, a UAV equipped with a road test terminal flew at varying altitudes (150m, 300m, 500m). The terminal was capable of receiving multiple SSB-Beams' RSRP and PCI values from the downlink of the serving cell BSs. Utilizing the PCI values, the terminals identified and matched the nearest BS with the corresponding PCI within the area. The samples collected from each region during the tests are presented in Table \ref{tableII}. 
%Furthermore, the visualization of the UAV's road test trajectory at 150m altitude is illustrated in Fig. 3.

% \begin{figure}[t]
%     \centering
%     \includegraphics[width=0.8\linewidth]{Dataset_Height_150_10_8.pdf}
%     \caption{The trajectory of UAV in these three areas at a height of 150 m.}
%     %\label{Fig Neural Network}
% \end{figure}

\subsection{Experiment Setting}
For the task of predicting coverage for BSs that have not been field-tested or are pending deployment, we have partitioned the entire dataset by binding each BS with its corresponding sparse samples. We then randomly divided the dataset centered around the BSs, as suggested in \cite{Eller_Svoboda_Rupp_2022}. This approach ensures physical isolation between different BSs, preventing data leakage due to the consideration of relative positional features as input.

In subsequent experiments, we define \(s\) as the BS sampling rate, which represents the proportion of the dataset allocated to the training set. The validation set is consistently set to occupy $10\%$ of the dataset, resulting in the test set comprising \(90\% - s\) of the dataset. Given the significant impact of dataset partitioning on experimental outcomes, we conduct 20 experiments for each training set proportion. These experiments maintain a fixed proportion but vary based on different random seeds, ranging from 0 to 19. To assess the performance of different algorithms, we employ the common coverage metric, i.e., MAE, on the test set.

\subsection{Ablation Study Result}

\subsubsection{Network Structure Comparison}
To validate the advanced nature of the proposed network structure, we established three benchmarks listed in Table \ref{tableIII} for ablation studies. In all benchmarks, the number of neurons in the hidden layers is set to 256, with the ReLU function serving as the activation function. For all benchmarks, the learning rate is set to 0.001, with a maximum of 100 epochs. Early stopping is implemented, terminating training when the loss on either the training or validation set decreases by no more than 1\% over ten consecutive epochs. For benchmark 1, the hidden layer counts for distance fading, frequency fading, and antenna gain networks are all 5. In contrast, benchmarks 2 and 3 feature 6 hidden layers. Compared to the benchmark proposed in this paper, benchmark 2 utilizes decoupled features as inputs and employs an MLP to perform predictions as depicted in Eq. (\ref{transformed Prediction task}). Conversely, benchmark 3 employs original features as inputs and uses an MLP to execute predictions according to Eq. (\ref{Prediction Task}).

\begin{figure}[t]
    \centering
    \includegraphics[width=0.8\linewidth]{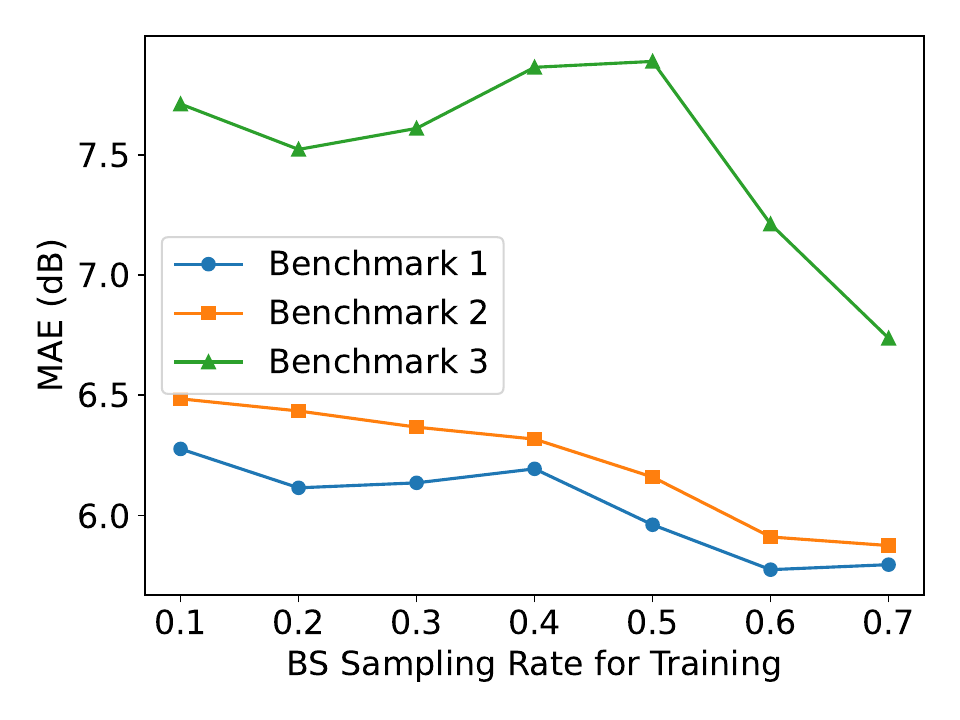}
    \caption{Average MAE for different algorithms across 20 random splits (random seeds 0-19) as a function of the BS sampling rate in ablation experiments.}
    \label{Fig_MAE_Ablation_Study}
\end{figure}

\begin{figure}[t]
    \centering
\includegraphics[width=0.8\linewidth]{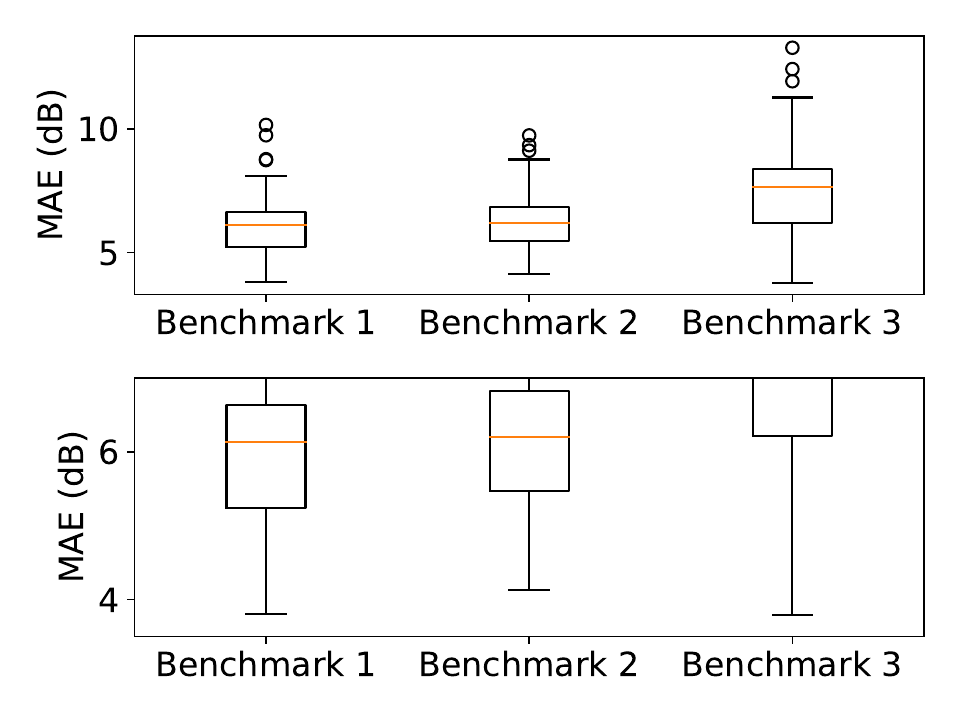}
    \caption{Boxplot of MAE performance for different algorithms across 20 random splits at sampling rates of \([10\%, 20\%, 30\%, 40\%, 50\%, 60\%, 70\%]\), totally 140 MAE data points per algorithm.}
    \label{Fig_Boxplot_Abalation_Study_1}
\end{figure}

\begin{figure}[t]
    \centering
    \includegraphics[width=0.8\linewidth]{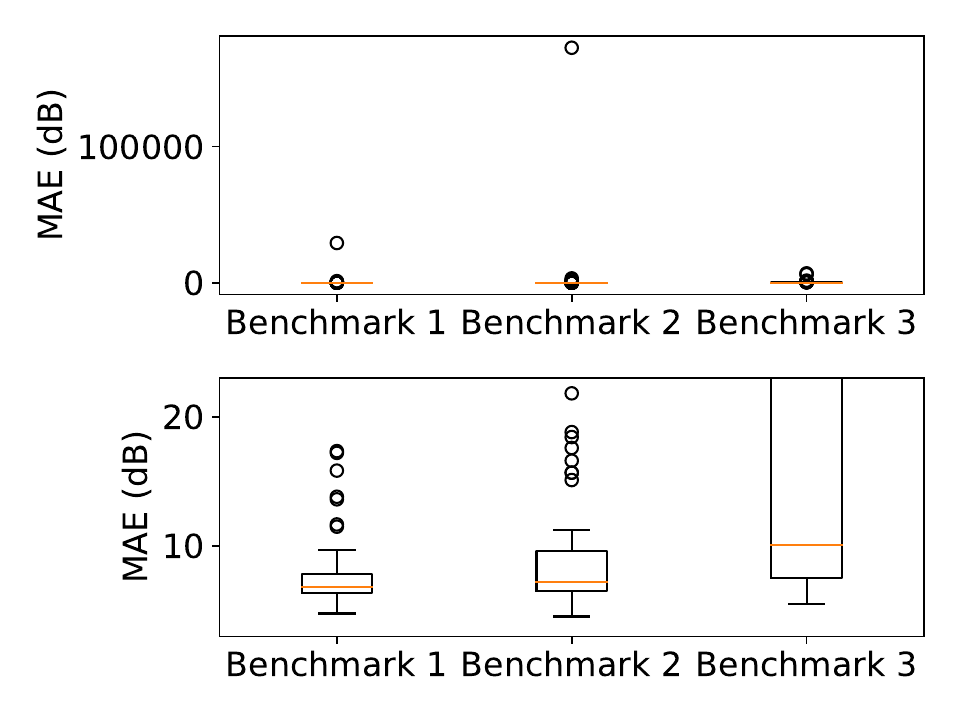}
    \caption{Boxplot of MAE performance for different algorithms across 20 random splits at sampling rates of \([1\%, 2\%, 3\%, 4\%, 5\%, 6\%]\), totally 120 MAE data points per algorithm.}
    \label{Fig_Boxplot_Abalation_Study_2}
\end{figure}

To demonstrate the effectiveness of the disentangled feature and network modules, we initially conducted experiments with 20 different dataset splits at sampling rates of \([10\%, 20\%, 30\%, 40\%, 50\%, 60\%, 70\%]\). For each sampling rate, we plotted the average MAE across benchmarks 1-3, as illustrated in Fig. \ref{Fig_MAE_Ablation_Study}. It is observable that the prediction errors for benchmarks 3, 2, and 1 progressively decrease, which validates the efficacy of the module design proposed in this study. Furthermore, as the sampling rate increases, the trends of the MAE curves for all three methods are consistently aligned, with overall performance gradually improving. This indicates that with an increase in training dataset size, the algorithms can learn a more precise distribution of the original dataset.

\subsubsection{Network Structure Robustness}
To demonstrate the robustness of the proposed algorithm, we visualized the boxplots of different algorithms during ablation experiments, as shown in Fig. \ref{Fig_Boxplot_Abalation_Study_1}. The different subplots within the figure visualize varying ranges of the Y-axis. Compared to benchmark 2, the proposed algorithm exhibits a lower mean and smaller fluctuation range. However, benchmark 2 achieves a lower maximum MAE, which may be attributed to the accidentality of outliers. The significantly higher MAE of benchmark 3 underscores the importance of feature compression design. Moreover, we conducted experiments with 20 random splits at low sampling rates of \([1\%, 2\%, 3\%, 4\%, 5\%, 6\%]\) to simulate early stages of low-altitude network planning with markedly insufficient training data, as well as scenarios where the algorithm is applied across different cities or operators with few matched AAU types between the test and training sets. As illustrated in Fig. \ref{Fig_Boxplot_Abalation_Study_2}, it is evident that even at such low training set proportions, the proposed method still achieves acceptable MAE performance (less than $7.5\ \mathrm{dB}$) on the test set for 75\% of the cases, significantly outperforming other algorithms. This further highlights the model's robustness in scenarios with low data availability.

\subsubsection{Importance of Right Model Knowledge}

\begin{table}[t]
\centering
\caption{Alternative Input for Different Networks in Proposed Method, where Networks 1,2 and 3 for Distance Fading, Frequency Fading, and Antenna Gain Networks, Respectively}

\begin{tabular}{cccc}
\\ \toprule

Model         & \begin{tabular}[c]{@{}c@{}}Network 1\end{tabular} & \begin{tabular}[c]{@{}c@{}}Network 2\end{tabular} & \begin{tabular}[c]{@{}c@{}}Network 3\end{tabular} \\ \midrule
Proposed      & \(D_n\)                                                                  & \(x^{(S)}_4\)                                                                   & \(\Delta \theta_n^{(H)}\), \(\Delta \theta_n^{(V)}\), \(\mathbf{x_4^{(S)}}\) \\
Wrong Model 1 & \(D_n\)                                                                  &\(\Delta \theta_n^{(H)}\)                                                                   &  \(x^{(S)}_4\),\(\Delta \theta_n^{(V)}\), \(\mathbf{x_4^{(S)}}\)                                                           \\
Wrong Model 2 & \(\Delta \theta_n^{(H)}\)                                                                  & \(x^{(S)}_4\)                                                                   & \(D_n\), \(\Delta \theta_n^{(V)}\), \(\mathbf{x_4^{(S)}}\)                                                           \\
Wrong Model 3 & \(\Delta \theta_n^{(H)}\)                                                                  & \(\Delta \theta_n^{(V)}\)                                                                   & \(D_n\), \(x^{(S)}_4\),  \(\mathbf{x_4^{(S)}}\)                         
\\ \bottomrule
\end{tabular}
\label{Table_True_Model_Study}
\end{table}
\begin{figure}[t]
    \centering
    \includegraphics[width=0.8\linewidth]{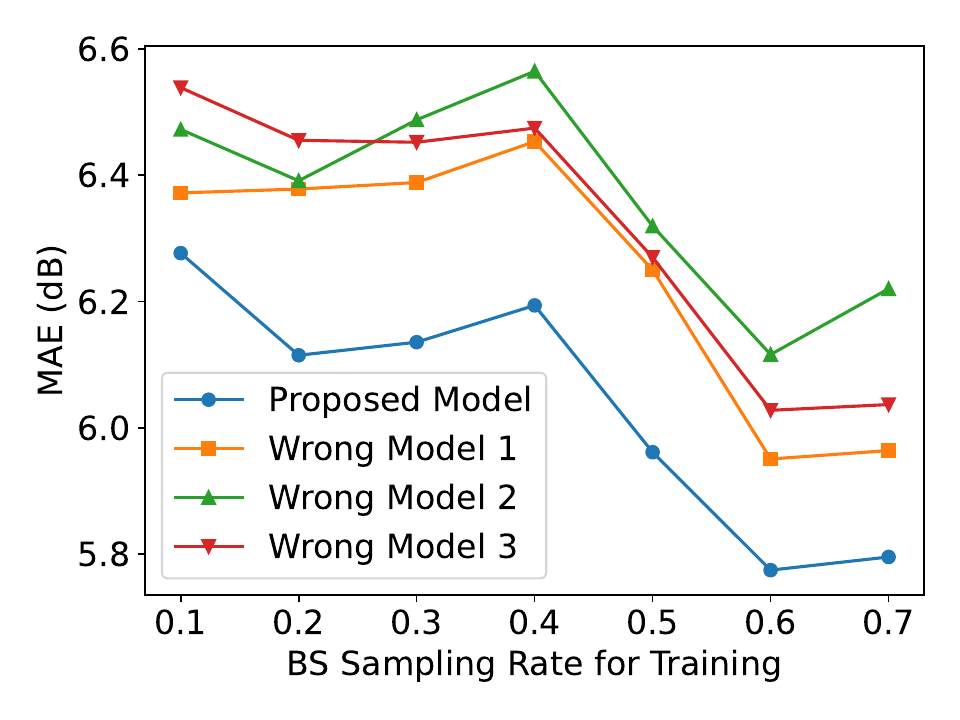}
    \caption{Average MAE for different model knowledge across 20 random splits (random seeds 0-19) as a function of the BS sampling rate.}
    \label{Fig_True_Model_Study}
\end{figure}
Based on the model described in Eq. (\ref{Eq_13}), we understand that the sets $\{\Delta \theta_n^{(H)}, \Delta \theta_n^{(V)}, \mathbf{x_4^{(S)}}\}$, $\{D_n\}$, and $\{x^{(S)}_4\}$ should be fed into distinct neural networks for learning representations. To investigate the impact of model correctness on learning performance, we considered alternative input combinations for the distance fading, frequency fading, and antenna gain networks, which are denoted as wrong models 1, 2, and 3, respectively, as shown in Table \ref{Table_True_Model_Study}. Furthermore, we have illustrated the MAE performance of these models in Fig. \ref{Fig_True_Model_Study}. It can be observed that the proposed model has an MAE error of less than other wrong models at all BS sampling rates. Compared with the best wrong model 1, the MAE of the proposed model decreases by about $0.2\ \mathrm{dB}$ on average. The observed superior performance of the proposed model highlights the importance of proper model knowledge guidance. Notably, wrong model 1 outperforms wrong models 2 and 3 slightly, which may be attributed to model 1 maintaining a separate network for the \(D_n\) input, assigning it a greater weight in impact on coverage prediction. 

It is important to note that not all possible combinations of model inputs were tested in this study. However, the experiments conducted sufficiently demonstrate the significance of correct knowledge in the considered model, Eq. (\ref{Eq_13}), for coverage prediction. It is worth mentioning that beyond the models of this discussion, the channel propagation model, Eq. (\ref{Eq_13}), may have more accurate prior knowledge structures, potentially leading to better feature preprocessing and combined representations. This could result in improved prediction performance, involving more precise channel propagation modeling, high-quality data-driven learning, and their trade-offs.

\subsubsection{Algorithm Generalization Capability}{Since we do not have infinite unknown data and cannot confirm the distribution of infinite data sets, the theoretical learning/generalization capability cannot be calculated, but we can verify the model's generalization capability experimentally. In the setting of the training set, validation set, and test set ratio of 6:2:2, 200 models were trained by randomly dividing the data set under 200 random seeds. We visualised the MAE CDF curves of these models in the training set, validation set, and test set.

According to the Fig. \ref{cdf}, the training set of all models is below 3.5 $\mathrm{dB}$, and the median is 2.733 $\mathrm{dB}$, indicating that the model has been basically fully trained. The CDF curves of validation MAE and test MAE are very close, and the median is 5.196 $\mathrm{dB}$ and 5.166 $\mathrm{dB}$, respectively, which are both low, indicating the reliability of the proposed method's Generalization capability. About 20\% of the models' test MAE is greater than 7 $\mathrm{dB}$, which may be caused by the large gap in data distribution between the training set and the test set. This problem can be improved by appropriately increasing the proportion of the verification set and training multiple models to vote.}
% \(\Delta \theta_n^{(H)}\), \(\Delta \theta_n^{(V)}\), \(\mathbf{x_4^{(S)}}\), \(D_n\), and \(x^{(S)}_4\)
\begin{figure}[t]
    \centering
    \includegraphics[width=0.78\linewidth]{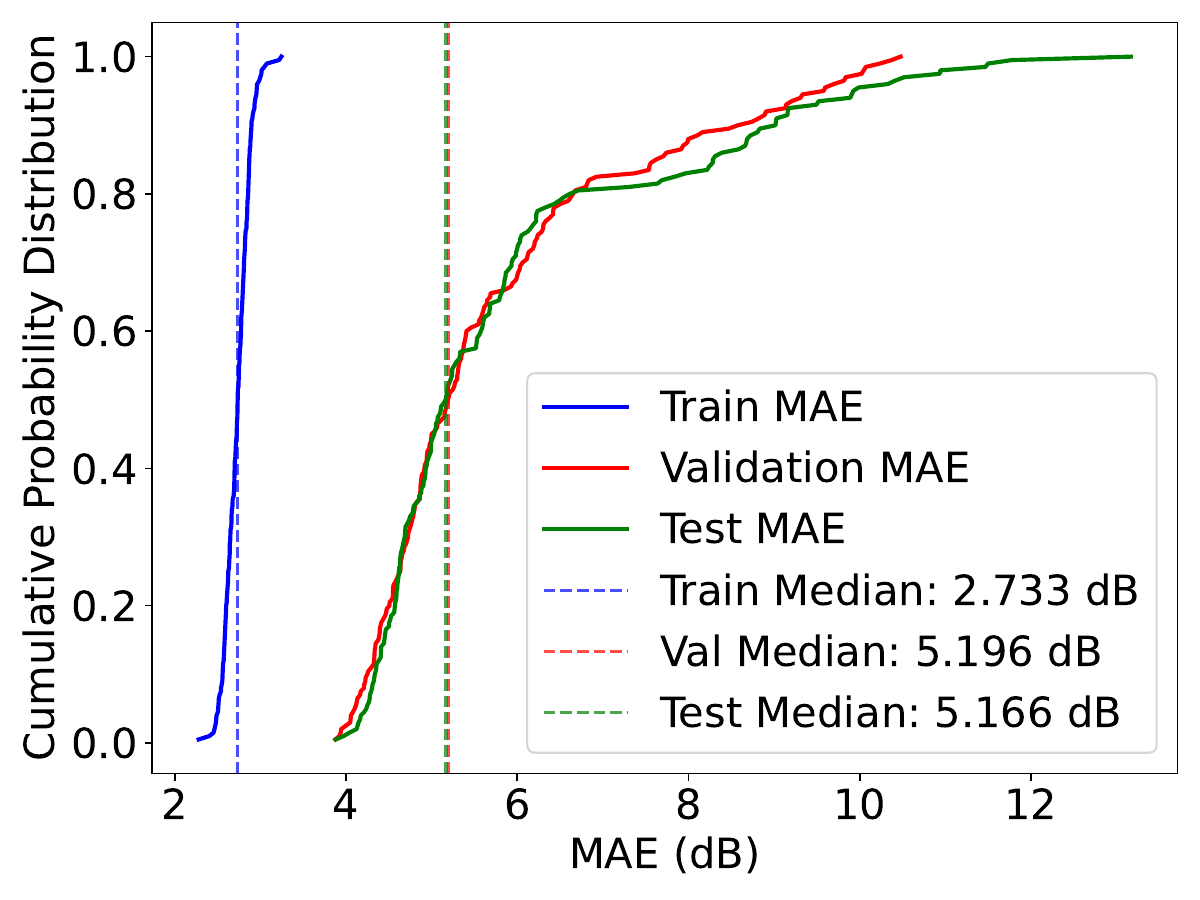}
    \caption{CDF curve of MAE metric for 200 trained neural networks across different dataset split seeds.}
    \label{cdf}
\end{figure}
\begin{figure}[t]
    \centering
    \includegraphics[width=0.8\linewidth]{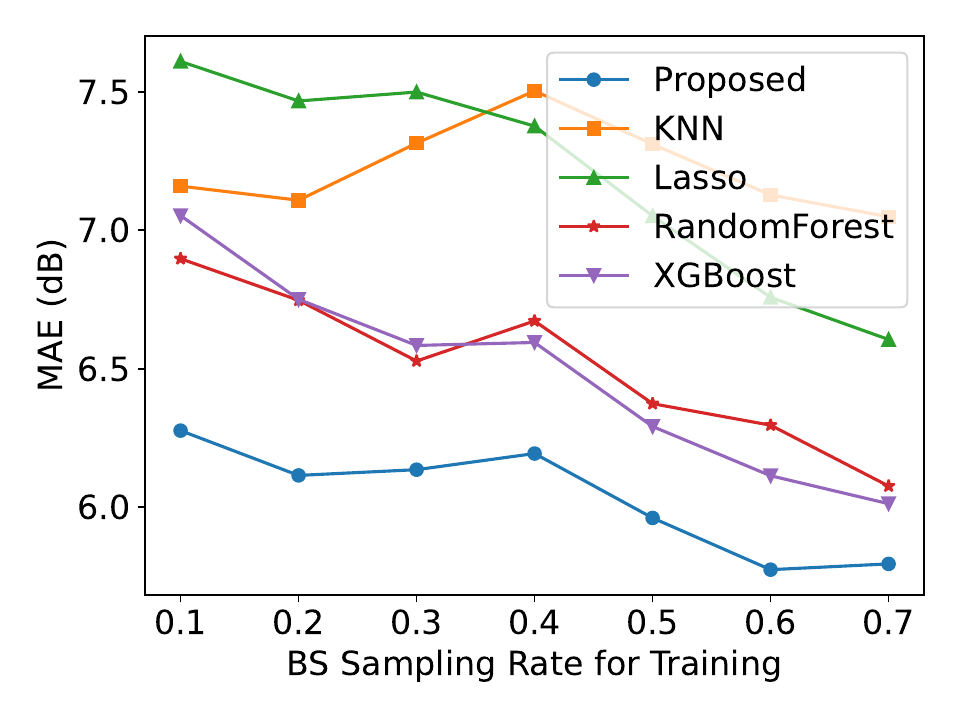}
    \caption{Average MAE for different algorithms across 20 random splits (random seeds 0-19) as a function of the BS sampling rate in ablation experiments.}
    \label{Fig_MAE_Alogrithm_Contrast}
\end{figure}

\begin{table*}[t]
\centering
\caption{Computational Complexity Comparison across Different Algorithms}
\label{tab:complexity}
\begin{tabular}{cccccc}
\toprule
Metric & Proposed & KNN & Lasso & RandomForest & XGBoost \\ 
\midrule
Model Size (MB) & 4.31 & 9.27 & 0.010 & 172.56 & 2.66 \\
Parameters & $\sim$1.1M & N/A\textsuperscript{*} & $\sim$2K & $\sim$45M & $\sim$680K \\
Inference Time (s)\textsuperscript{**} & $<$0.5 & 1.2 & 0.1 & 3.8 & 0.8 \\
Memory Usage (GB) & 0.8 & 2.1 & 0.05 & 4.2 & 1.3 \\
\bottomrule
\multicolumn{6}{l}{\footnotesize \textsuperscript{*}KNN is instance-based without explicit parameters} \\
\multicolumn{6}{l}{\footnotesize \textsuperscript{**}For 2,500 base station predictions on RTX 4090}
\end{tabular}
\end{table*}

\begin{table}[t]
\centering
\caption{\centering{BS Operational Parameters Setting}}
\begin{tabular}{ccc}
\toprule
Parameters         & BS 1        & BS 2         \\ \midrule
Antenna Height     & 39.86m      & 38.0m        \\
Coverage Scenario  & SCENARIO\_0 & SCENARIO\_21 \\
Num of Channels    & 32T32R      & 64T64R       \\
AAU Type           & AAU5336e    & AAU5636      \\
Mechanical down tilt    & 9.72°       & 0°           \\
Horizontal Azimuth            & 359°        & 240°         \\
Relative SSB Power & 0dB         & 2.01dB     \\ \bottomrule 
\end{tabular}
\label{Table_V}
\end{table}
\begin{figure}[t]
\centering
    \includegraphics[width=0.9\linewidth]{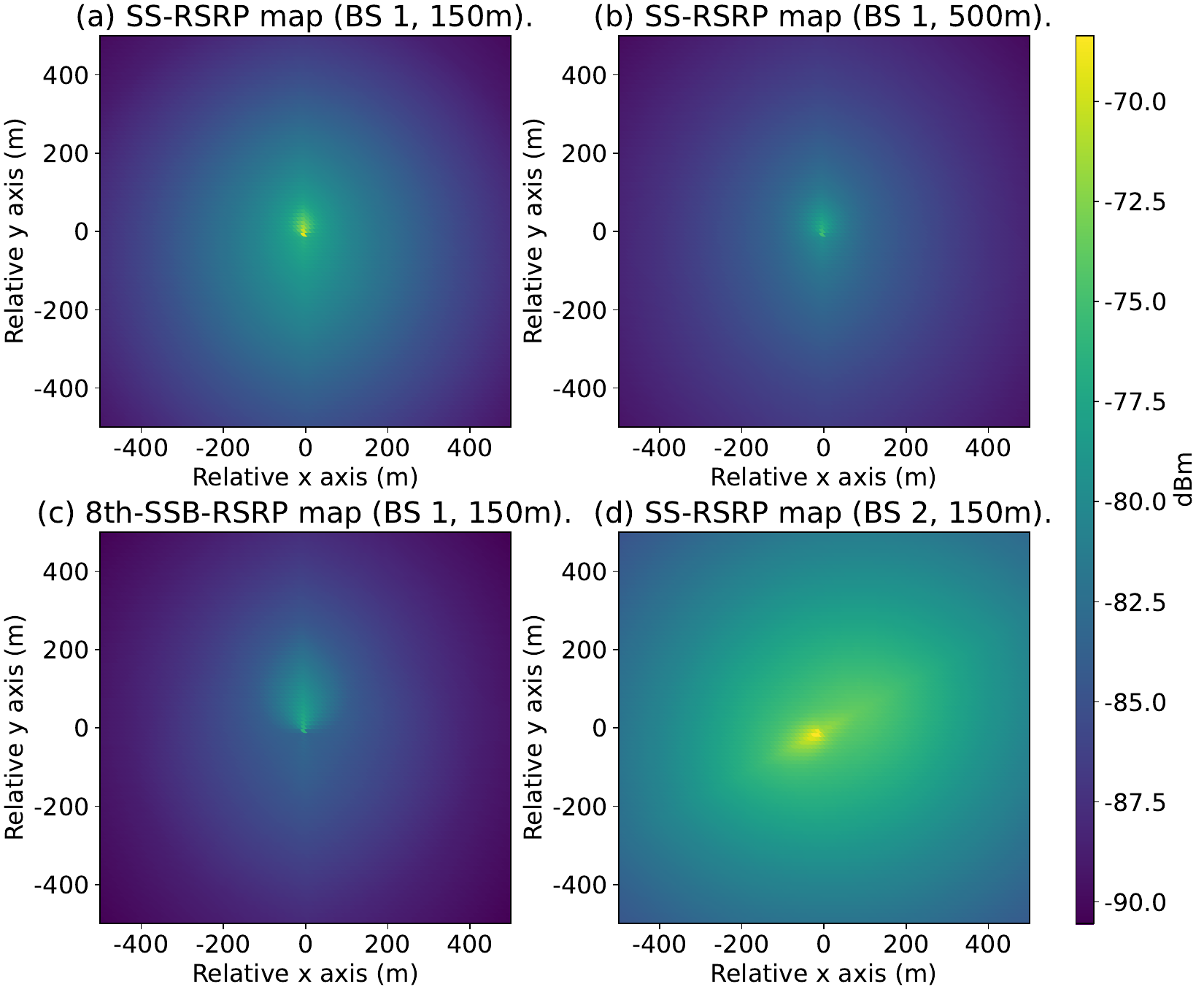}
\caption{The prediction coverage visualisation results for different operational parameters under different conditions.}
\label{fig_Single_BS}
\end{figure}
\begin{figure}[t]
\centering
\centering
\includegraphics[width=0.98\linewidth]{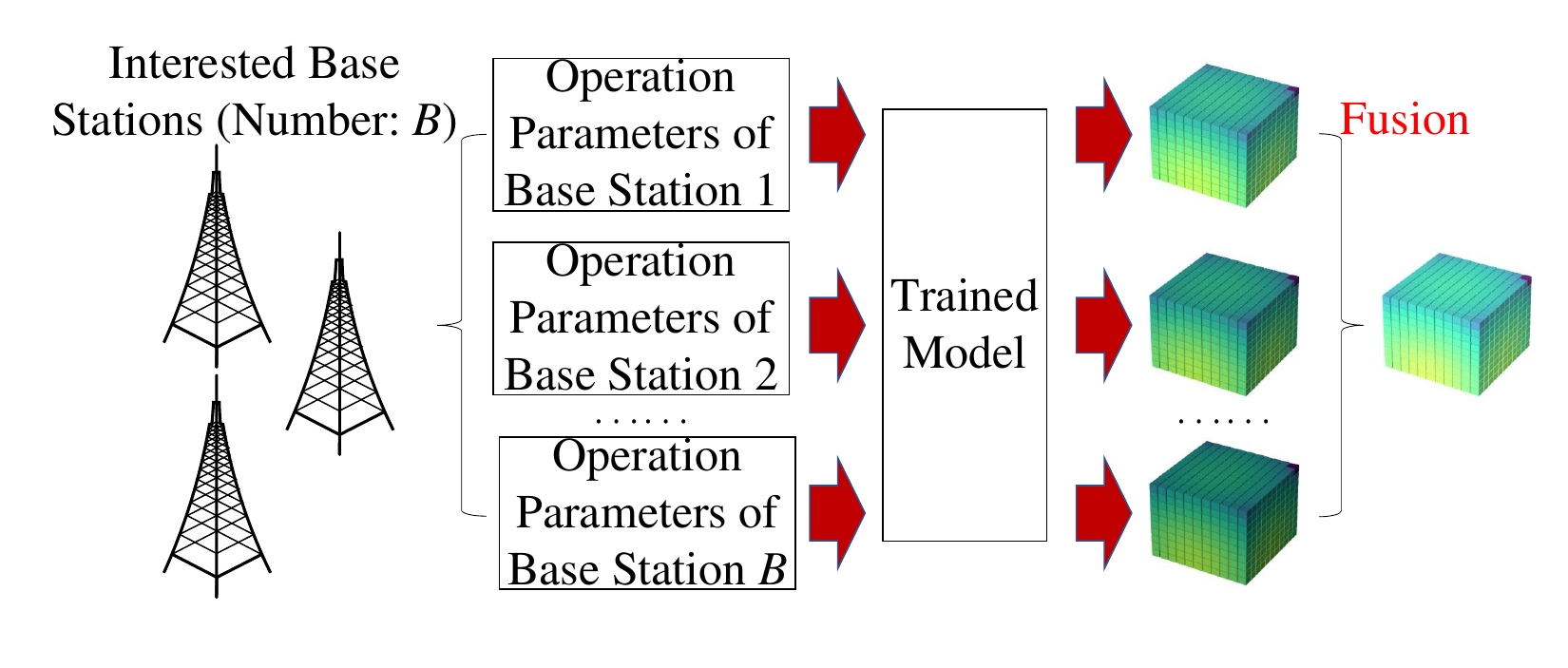}
\caption{Multi-BS coverage prediction framework. Operational parameters from $B$ interested base stations are input to the trained model to generate individual coverage predictions, which are then fused to produce the integrated coverage map.}
\label{New_fig1}
\end{figure}
\begin{figure}[t]
    \centering
    \includegraphics[width=0.8\linewidth]{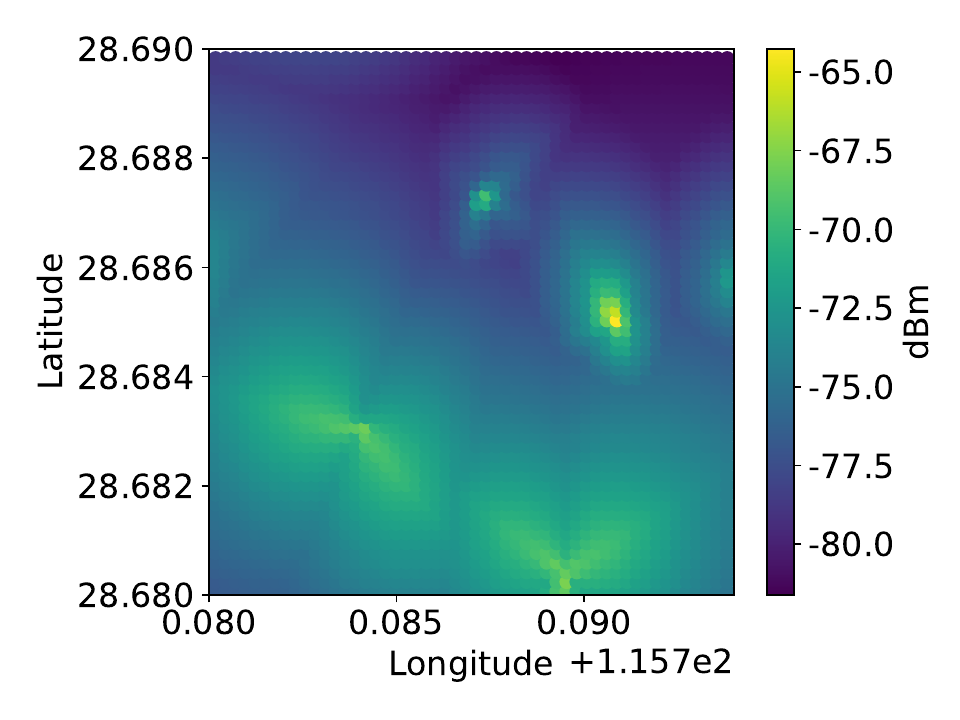}
    \caption{SS-RSRP coverage map visualisation in Xinyue Lake before operation parameters optimisation.}
    \label{Fig_Multi_BS_Coverage}
\end{figure}

\begin{figure*}[t]
    \centering
    \includegraphics[width=0.98\linewidth]{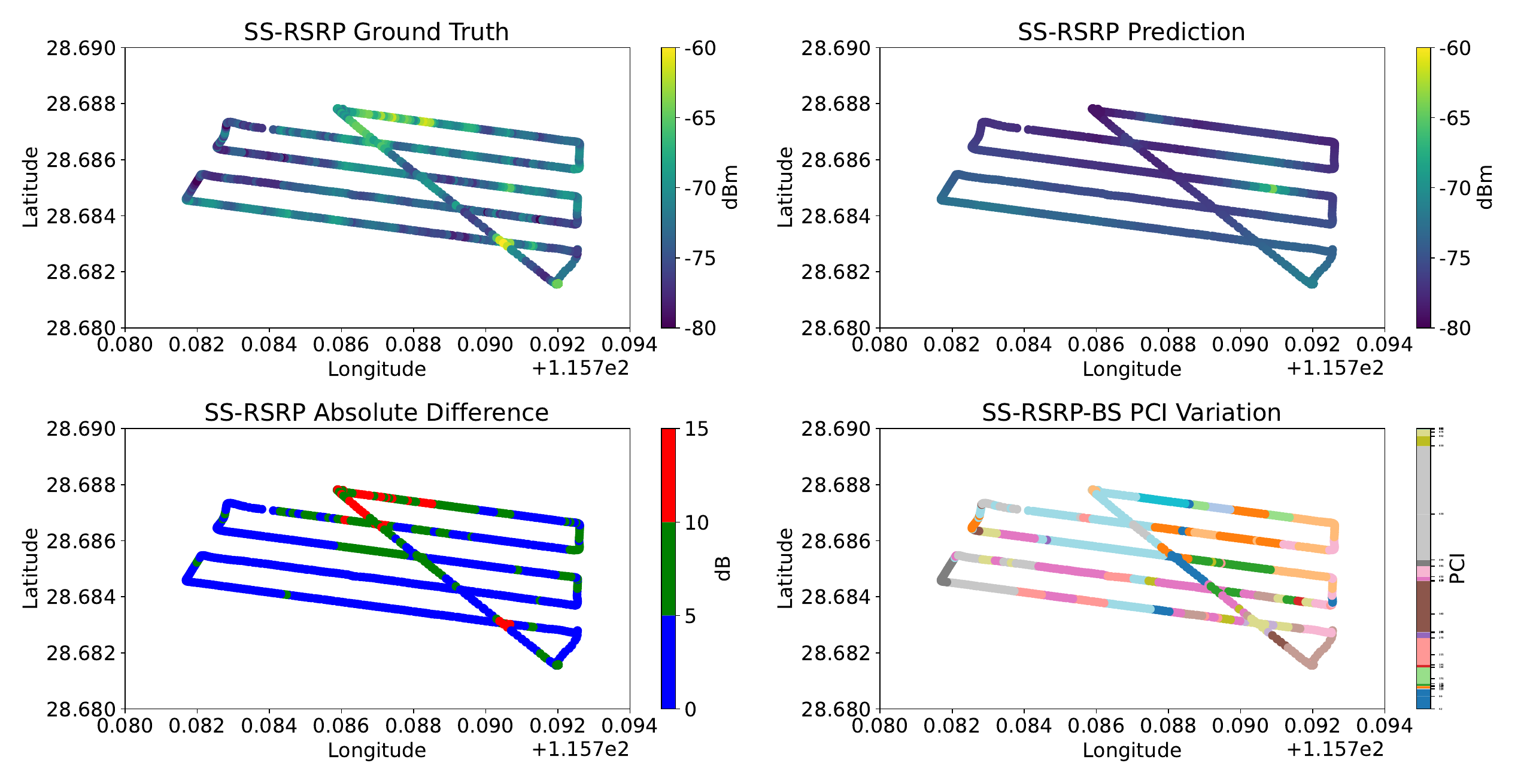}
    \caption{Comparison of predicted and actual coverage at sampling locations in Xinyue Lake before operational parameters optimisation (top left: actual SS-RSRP coverage at sampling locations, top right: predicted SS-RSRP coverage at sampling locations, bottom left: absolute difference between actual and predicted SS-RSRP at sampling locations, bottom right: UAV PCI switching during flight).}
    \label{Fig_Sampling_Comparison}
\end{figure*}

\begin{figure}[t]
    \centering
    \includegraphics[width=0.8\linewidth]{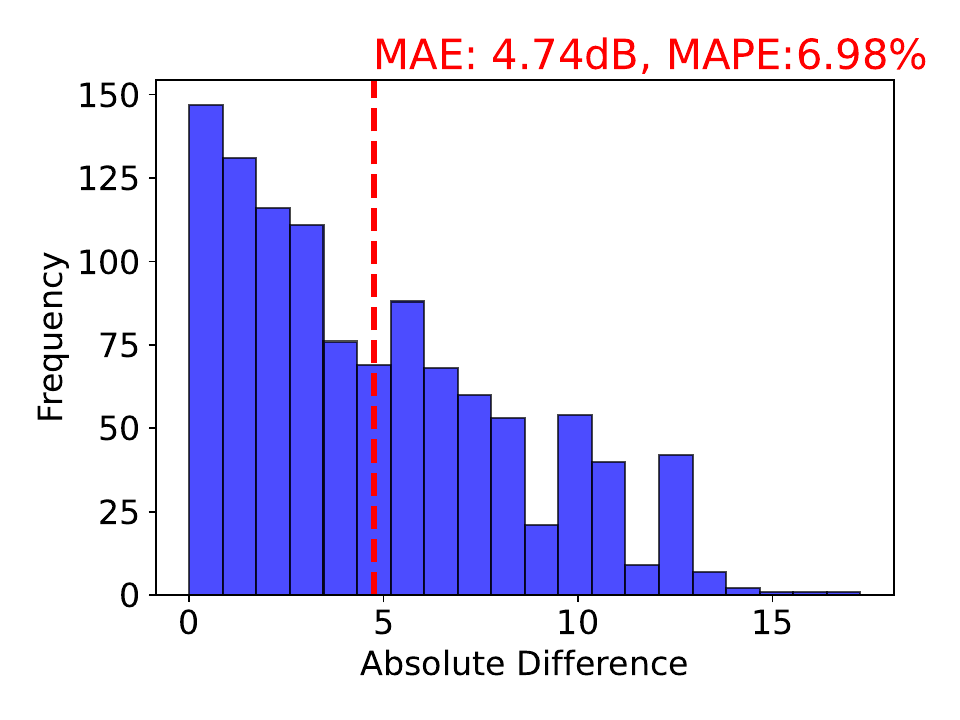}
    \caption{The distribution of truth-prediction absolute error for all samples in Xinyue Lake before operation parameters optimisation.}
    \label{Fig_MAE_Distribution_1}
\end{figure}

\begin{figure}[t]
    \centering
    \includegraphics[width=0.8\linewidth]{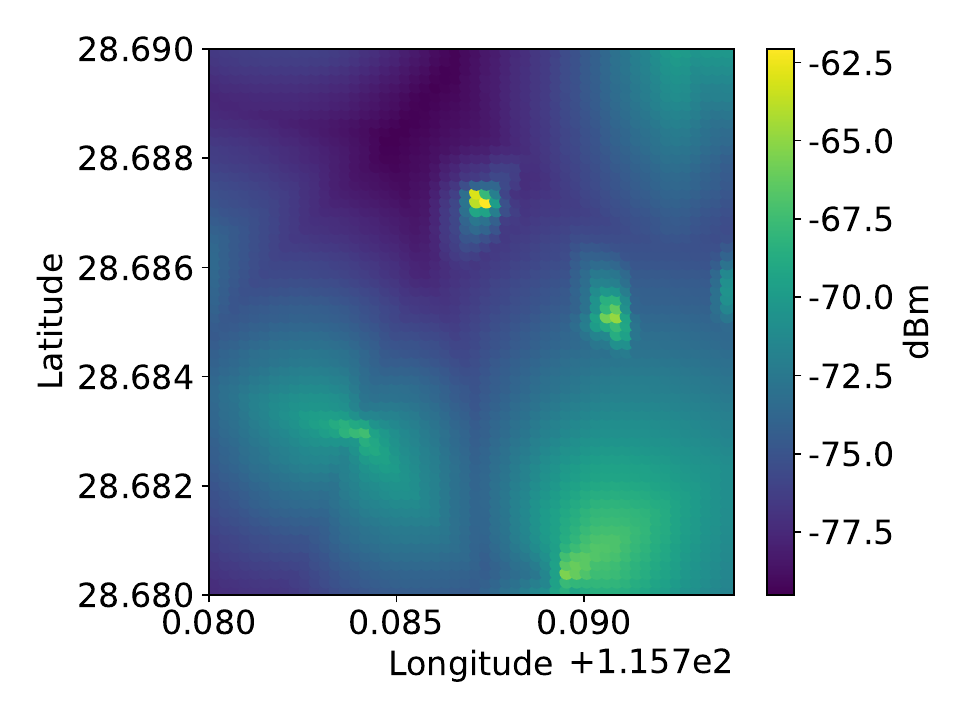}
    \caption{SS-RSRP coverage map visualisation in Xinyue Lake after operation parameters optimisation.}
    \label{Fig_Multi_BS_Coverage_2}
\end{figure}

\begin{figure}[t]
    \centering
    \includegraphics[width=0.85\linewidth]{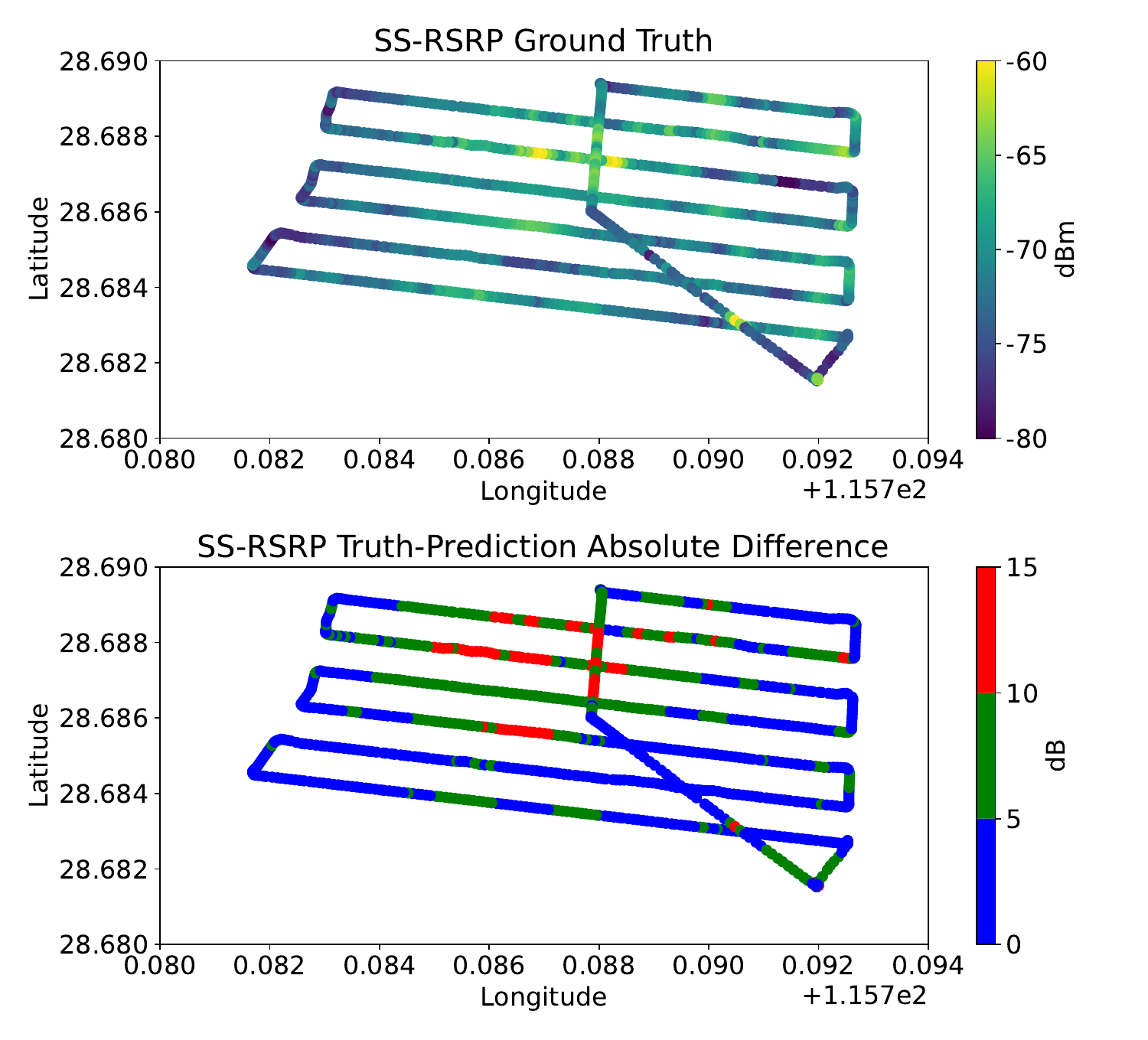}
    \caption{Comparison of predicted and actual coverage at sampling locations in Xinyue Lake after operational parameters optimisation (top: actual SS-RSRP coverage at sampling locations, bottom: absolute difference between actual and predicted SS-RSRP at sampling locations).}
    \label{Fig_Sampling_Comparison_2}
\end{figure}

\begin{figure}[t]
    \centering
    \includegraphics[width=0.8\linewidth]{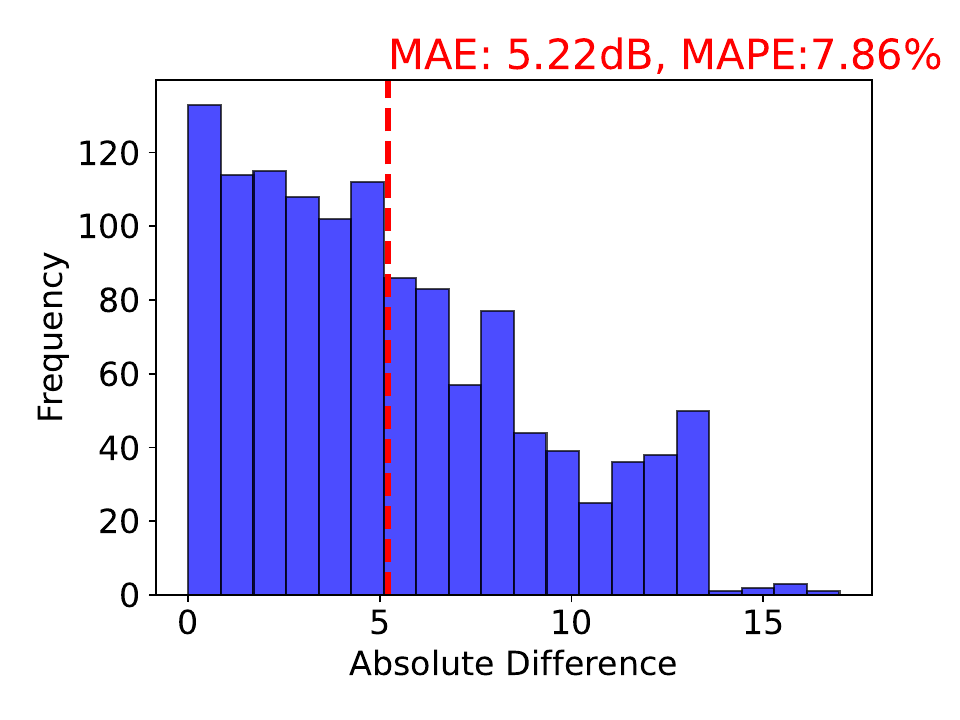}
    \caption{The distribution of truth-prediction absolute error for all samples in Xinyue Lake after operation parameters optimization.}
    \label{Fig_MAE_Distribution_2}
\end{figure}

\begin{figure}[t]
    \centering
    \includegraphics[width=0.8\linewidth]{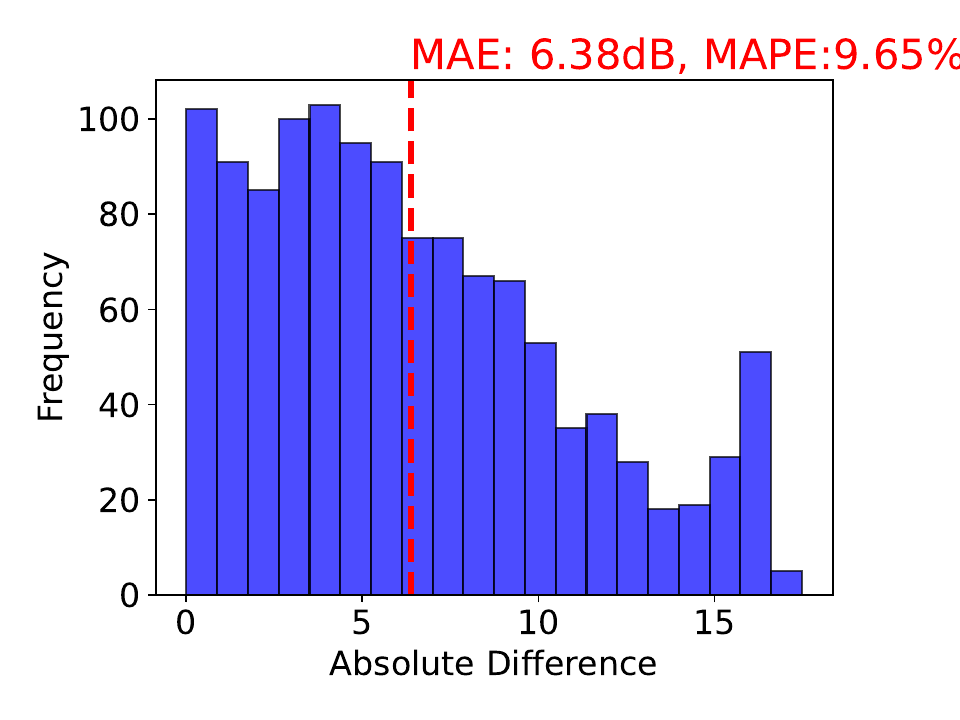}
    \caption{The distribution of truth-prediction absolute error for all samples in Xinyue Lake after operation parameters optimization (if the prediction is sampled from surface coverage in Fig. \ref{Fig_Multi_BS_Coverage}).}
    \label{Fig_MAE_Distribution_3}
\end{figure}

\subsection{Algorithm Comparison}
In addition to the ablation study, the proposed method was compared with several traditional machine learning methods, including KNN\cite{mohammadjafari2020machine} (with the number of neighbors set to 50), Lasso \cite{moreta2019prediction} (with a regularization strength of 1), Random Forest\cite{9344905} (with 100 trees), and XGBoost\cite{wang2020lte} (using mean squared error as the objective). As shown in Fig. \ref{Fig_MAE_Alogrithm_Contrast}, the MAE of different algorithms was plotted against varying BS sampling rates. The results indicate that the proposed method consistently outperforms traditional methods across various BS sampling rates, achieving an average MAE improvement of approximately $0.5\ \mathrm{dB}$. This further demonstrates the superiority of the proposed method.

{The superior performance stems from how the proposed framework addresses fundamental limitations of traditional algorithms in low-altitude scenarios. KNN's degradation with increased sampling rates occurs because it relies on Euclidean distance in a feature space with imbalanced dimensionality—high-dimensional BS parameters are sparsely sampled while low-dimensional locations vary densely. In low-altitude networks where beam patterns dominate coverage, nearby neighbours in feature space may belong to BSs with drastically different antenna configurations, causing poor predictions. Linear methods like Lasso fail because they cannot capture the non-linear interactions between distance fading, frequency-dependent attenuation, and directional antenna gain described in Eq. (\ref{Eq.9}) and Eq. (\ref{Eq_13}), particularly the angular transformations in Eq. (\ref{arctan2})-(\ref{Eq_7}). Tree-based ensembles (Random Forest, XGBoost), while non-linear, process features independently through recursive partitioning without leveraging propagation physics. 

In contrast, our disentangled representation network explicitly decomposes predictions into physically meaningful components through dedicated sub-networks for distance fading, frequency fading, and antenna gain, with outputs combined additively per Eq. (\ref{Eq_13}). The expert knowledge-based feature compression (Eq. (\ref{arctan2})-(\ref{Eq_8})) transforms imbalanced features into balanced representations, while the constrained architecture reduces overfitting risk (Propositions \ref{prop1}-\ref{prop2})—critical when individual BS data is sparse. This demonstrates the framework's superiority in exploiting domain knowledge for improved accuracy and generalisation in low-altitude scenarios.}
\subsection{Computational Complexity Analysis}

{To assess the practical applicability of our proposed method, we conduct a comprehensive computational complexity analysis from both theoretical and empirical perspectives.}

\subsubsection{Theoretical Complexity Analysis}

{We first analyze the theoretical computational advantages of our disentangled representation network compared to conventional MLP architectures. Let $V_{\text{net}}$ and $D_{\text{net}}$ denote the hidden layer width and depth, respectively.}

{For our disentangled representation network, the parameter count in the hidden layer is
\begin{equation}
N_{\text{proposed}} = D_{\text{net}} \times \left[\left(\frac{V_{\text{net}}}{3}\right)^2+\frac{V_{\text{net}}}{3}\right]\times 3=\frac{V_{\text{net}}^2D_{\text{net}}}{3}+V_{\text{net}}D_{\text{net}}.    
\end{equation}}
{In contrast, a conventional MLP with equivalent representation capacity requires
\begin{equation}
N_{\text{MLP}} =V_{\text{net}}^2D_{\text{net}}+V_{\text{net}}D_{\text{net}}.
\end{equation}}
{This corresponds to a parameter reduction of
\begin{equation}
\frac{N_{\text{MLP}} - N_{\text{proposed}}}{N_{\text{MLP}}} = \frac{2V_{\text{net}}^2D_{\text{net}}}{3(V_{\text{net}}^2D_{\text{net}}+V_{\text{net}}D_{\text{net}})} \approx 66\%.
\end{equation}}
{This theoretical analysis demonstrates that by strategically incorporating communication domain knowledge through decoupled spatial processing, our architecture achieves substantial parameter efficiency while maintaining representational capability.}

\subsubsection{Empirical Complexity Comparison}

{To validate the practical computational advantages, we compare our method with established baseline algorithms across multiple complexity metrics. Table \ref{tab:complexity} presents comprehensive comparisons including model size, parameter count, and inference time.}

{The empirical analysis reveals that our proposed method achieves an optimal balance between model complexity and computational efficiency. While Lasso demonstrates the smallest model size due to its linear nature, it lacks the representational capacity for complex spatial relationships in LANC prediction. Conversely, RandomForest achieves competitive accuracy but at the cost of significantly larger model size and inference time. Our approach maintains moderate model size (4.31 MB) while delivering sub-second inference for large-scale predictions, making it practically viable for real-time network planning applications.}

\subsection{Applications and Visualizations}
Beyond basic experimental evaluations, which include ablation studies and algorithmic comparisons with machine learning methods, we have also conducted extensive application validations on additional real-world datasets to demonstrate the efficacy of the proposed method.
\subsubsection{Low-altitude Coverage Prediction for An Individual BS}
The effective prediction of low-altitude coverage for individual base stations constitutes a crucial foundation for network coverage. To illustrate the capability of the method in predicting the coverage of a single base station, we have visualized the coverage prediction of the proposed method under different operational parameters as shown in Table \ref{Table_V}\footnote{We only show the different operational parameters. The relative SSB power in Table \ref{Table_V} is calculated as shown in Eq. (\ref{Eq_10}). In addition, SCENARIO\_0 represents a common ground-based wide beam configuration, whereas SCENARIO\_21 is configured for aerial beamforming.}. Given the extremely limited sample size for individual base stations, it is challenging to quantitatively assess the prediction performance. Therefore, we provide a brief analysis of different coverage prediction results in conjunction with the operational parameters.

From the comparisons in Fig. \ref{fig_Single_BS}(a) with \ref{fig_Single_BS}(b) and \ref{fig_Single_BS}(c), it is evident that the predicted SS-RSRP coverage is better than the predicted 8th-SSB-RSRP coverage. The coverage at a height of 500 meters is slightly weaker than that at 150 meters. From the comparisons in Fig. \ref{fig_Single_BS}(a) with \ref{fig_Single_BS}(d), the proposed method effectively captures changes in the horizontal angles of BS parameters and demonstrates better overall low-altitude coverage performance under aerial beam configurations.
\subsubsection{Multi-BS Coverage Prediction before Operational Parameters Optimization}
To identify low-altitude coverage holes, it is often necessary to predict the joint coverage of multiple BSs within a specific area. {As illustrated in Fig. \ref{New_fig1}, the proposed method first predicts the coverage of individual BSs using their respective operational parameters.} Then, the coverage predictions of multiple BSs are combined and rasterized, with the maximum predicted RSRP in each grid selected as the predicted SS-RSRP coverage for that grid. Using the trained model, we selected an area in the Xinyue Lake region of Jiangxi Province, involving 40 BSs. For each BS, the coverage at a height of 120 meters and a radius of 2000 meters (with a resolution of 10m × 10m) was predicted and combined. To ensure no data leakage, we removed all samples in the training set with longitude in [115.78, 115.794] and latitude in [28.68, 28.69] and trained the proposed model again, as the training set contained some low-altitude UAV measurement samples from the Xinyue Lake region. The visualization results can be shown in Fig. \ref{Fig_Multi_BS_Coverage}.

To evaluate the accuracy of the predictions, additional low-altitude measurements were conducted at a height of 120 meters. The predicted results were sampled at the same locations as the test points and compared with the new measurement results. The visualization results are shown in Fig. \ref{Fig_Sampling_Comparison}.
It can be observed that the predicted SS-RSRP closely matches the measured RSRP, with the absolute difference between predicted and measured SS-RSRP values being less than $5\ \mathrm{dB}$ for most points. Based on the UAV measurement data of PCI changes for connected BSs, it is evident that larger prediction errors primarily occur at sampling points where BS handovers take place. These sudden handovers introduce prediction errors, which is reasonable. To reduce such errors, future works of coverage modelling and prediction could incorporate actual BS handover parameters and cell relationships.

To provide a more precise prediction error, we visualized the distribution of the absolute differences between predicted and measured SS-RSRP values across sampling locations in Fig. \ref{Fig_MAE_Distribution_1}. The results show over 90\% of the sampling have an error of less than $8\ \mathrm{dB}$, and the MAE across all sampling locations is 4.74 dB. This further demonstrates the accuracy and practicality of the proposed method.

\subsubsection{Multi-BS Coverage Prediction after Operational Parameters Optimization}
To enhance the quality of low-altitude coverage in the area, we optimized the operational parameters of 20 BSs in the region following the method described in \cite{zhang2023physics}. Based on the optimized parameters, the coverage at a height of 120 meters was re-predicted and re-measured\footnote{We take a new round of low-altitude trajectory collection at 120 meters.}, as shown in Fig. \ref{Fig_Multi_BS_Coverage_2} and \ref{Fig_Sampling_Comparison_2}. Moreover, the coverage prediction error distribution of all samples was visualized in Fig. \ref{Fig_MAE_Distribution_2}, demonstrating that the proposed algorithm maintained excellent MAE performance. To more effectively demonstrate the algorithm's effectiveness in adjusting network operational parameters for coverage prediction, we sampled based on the surface coverage in Fig. \ref{Fig_Multi_BS_Coverage} and calculated the MAE for the adjusted sampling points, as illustrated in Fig. \ref{Fig_MAE_Distribution_3}. A comparison between Fig. \ref{Fig_MAE_Distribution_2} and Fig. \ref{Fig_MAE_Distribution_3} reveals that the accuracy of surface coverage prediction based on adjusted operational parameters outperforms the accuracy before adjustment, showing an approximate 18\% performance improvement. This indicates the proposed algorithm's sensitivity to parameter adjustments, allowing for highly accurate coverage evaluation after network optimization.

\subsubsection{Multi-BS Coverage Prediction in Ganzhou without AAU Type of BSs}
When predicting low-altitude coverage across different operators, cities, or equipment manufacturers, the AAU type of the BSs to be predicted is often unknown during model training. To verify the performance of the proposed method in this case, we selected BSs along a low-altitude flight path in Ganzhou, Jiangxi Province, as the subjects for prediction. Given that the BSs in Ganzhou are all manufactured by ZTE, while those in the training data are from Huawei, all the AAU types of the BSs to be predicted are considered as a new category. Consequently, we excluded the AAU type feature $x_1^{(S)}$ from the input features and retrained a model using the training set. 

{The maintained prediction accuracy without AAU type information can be attributed to two key factors. First, as categorised in Table \ref{table1}, the AAU type primarily determines the availability of antenna channels, which is already explicitly provided as a separate input feature ($x_2^{(S)}$, number of channels). The coverage scenario feature ($x_3^{(S)}$) further specifies the beam antenna patterns configured for each scenario at the corresponding carrier frequency ($x_4^{(S)}$). Therefore, the critical beam characteristics, including beam count, beam width, and beamforming capabilities, are still implicitly captured through the combination of these remaining features without explicit AAU manufacturer-specific information. Second, the disentangled representation framework learns to extract antenna gain patterns through the antenna gain network, which takes $\{\Delta\theta_n^{(H)}, \Delta\theta_n^{(V)}, x_2^{(S)},x_3^{(S)}\}$ as inputs. This network can generalize across different AAU types by learning the functional relationship between relative angles and gain, conditioned on the remaining static beam characteristics. The model essentially learns that BSs with similar channel configurations and coverage scenarios exhibit similar beam patterns, regardless of the specific AAU manufacturer.}
\begin{figure}[t]
    \centering
    \includegraphics[width=0.8\linewidth]{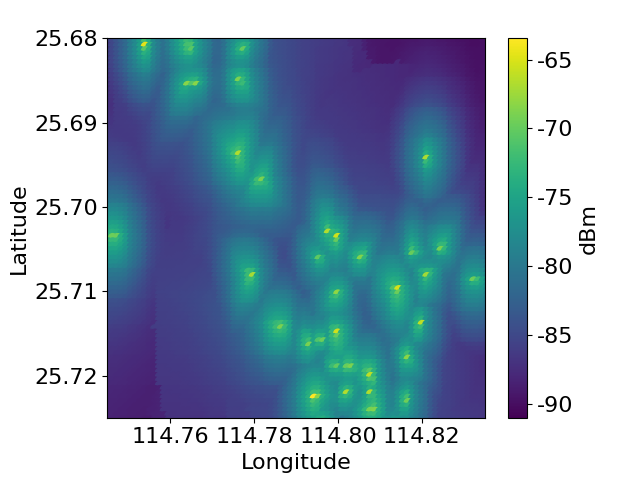}
    \caption{SS-RSRP coverage map visualization in Ganzhou.}
    \label{Fig_Multi_BS_Coverage_3}
\end{figure}

\begin{figure}[t]
    \centering
    \includegraphics[width=0.85\linewidth]{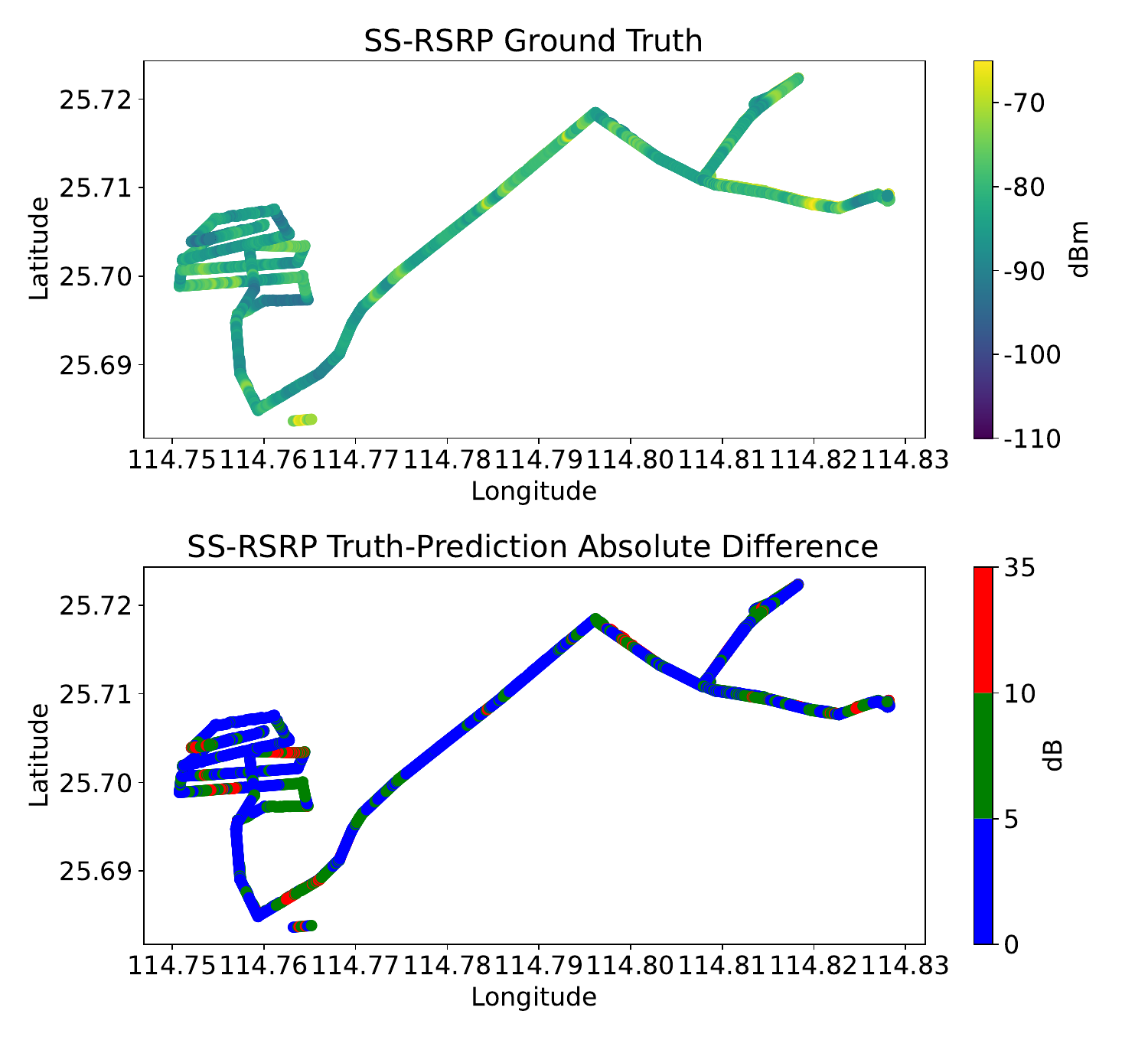}
    \caption{Comparison of predicted and actual coverage at sampling locations in Ganzhou (top: actual SS-RSRP coverage at sampling locations, bottom: the absolute difference between actual and predicted SS-RSRP at sampling locations).}
    \label{Fig_Sampling_Comparison_3}
\end{figure}

\begin{figure}[t]
    \centering
    \includegraphics[width=0.8\linewidth]{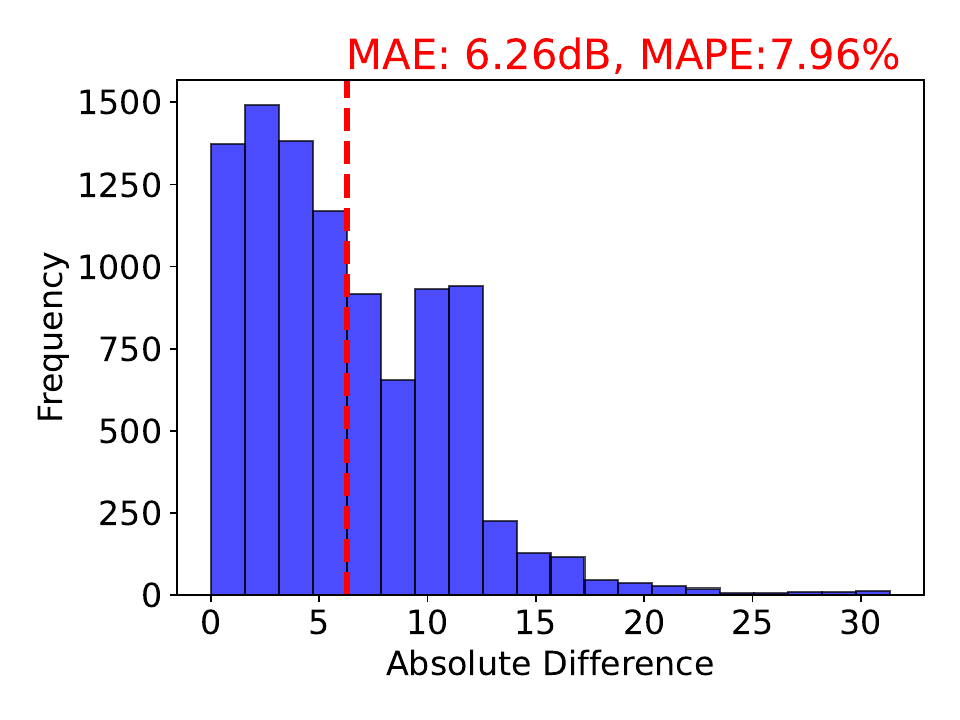}
    \caption{The distribution of truth-prediction error for all samples in Ganzhou.}
    \label{Fig_MAE_Distribution_4}
\end{figure}
By inputting the operational parameters of the BSs (a total of 60 BSs) along the flight path, we predicted the low-altitude coverage at 150 meters altitude. Specifically, we predicted the RSRP for each of the 8 beams of every BS, followed by stacking and rasterization, with the maximum predicted RSRP in a grid defined as the grid SS-RSRP. The predicted surface coverage is shown in Fig. \ref{Fig_Multi_BS_Coverage_3}. Further sampling of the surface coverage, as illustrated in Fig. \ref{Fig_Sampling_Comparison_3}, reveals that the majority of the sampling points along the flight path have a predicted error of less than $5\ \mathrm{dB}$. To quantify the error more accurately, as shown in Fig. \ref{Fig_MAE_Distribution_4}, we visualized the distribution of the absolute differences between predicted SS-RSRP and real sampled SS-RSRP. The method proposed in this study achieved an MAE of $6.26\ \mathrm{dB}$ and a MAPE of 7.96\%. {The slight performance degradation compared to the in-domain prediction (MAE of $4.74\ \mathrm{dB}$ in Fig. \ref{Fig_MAE_Distribution_1}) is expected due to potential manufacturer-specific differences in beam pattern implementations. Nevertheless, this result demonstrates the reliability of the method and its transferability across different operator scenarios, even when complete BS hardware information is unavailable.}

\section{Conclusion}\label{Section: Conclusion}
This research has made significant strides in addressing the intricate challenge of low-altitude network coverage (LANC) prediction, a cornerstone for the burgeoning aerial economy. Recognizing the impediment faced by equipment manufacturers in accessing BS antenna beam patterns, we innovatively leveraged BS operational parameters for the prediction of LANC, adopting a data-driven model that marks a substantial advancement over traditional approaches.

Our research has introduced two pioneering techniques: feature compression based on expert knowledge and a model-guided disentangled representation network. The former not only achieves data compression but also addresses the issue of feature imbalance, while the latter, by integrating a signal propagation model, mitigates reliance on extensive data and reduces the risk of overfitting. Both of the techniques have demonstrated their efficacy in our experimental validation based on real-world data, {significantly enhancing the accuracy of low-altitude coverage prediction.}

{However, we acknowledge the limitations of our study, particularly the insufficient consideration of environmental obstructions affecting signal propagation in low-altitude networks. According to related works, incorporating environmental information could potentially improve prediction accuracy by approximately 0.5-1 $\mathrm{dB}$ RMSE for radio map reconstruction problems \cite{10682510}, and by 1-2 $\mathrm{dB}$ RMSE for radio map prediction without sparse sampling \cite{9954403}, where the larger gain is attributed to the absence of measurement data. These quantitative insights highlight the potential impact of environmental factors on prediction performance. Additionally, there remains an opportunity to integrate more precise propagation models into our neural network architecture.}

{Future endeavours will focus on exploiting large-scale data and large language models to refine LANC predictions. We also aim to incorporate detailed environmental data and BS beamforming patterns into our models, thus improving the precision and resilience of LANC while addressing the identified limitations.}

% if have a single appendix:
%\appendix[Proof of the Zonklar Equations]
% or
%\appendix  % for no appendix heading
% do not use \section anymore after \appendix, only \section*
% is possibly needed

% use appendices with more than one appendix
% then use \section to start each appendix
% you must declare a \section before using any
% \subsection or using \label (\appendices by itself
% starts a section numbered zero.)
%

% \appendices
% \section{Proof of the First Zonklar Equation}
% Appendix one text goes here.

% % you can choose not to have a title for an appendix
% % if you want by leaving the argument blank
% \section{}
% Appendix two text goes here.

% % use section* for acknowledgment
% \section*{Acknowledgment}

% The authors would like to thank...

% Can use something like this to put references on a page
% by themselves when using endfloat and the captionsoff option.
\ifCLASSOPTIONcaptionsoff
  \newpage
\fi

% trigger a \newpage just before the given reference
% number - used to balance the columns on the last page
% adjust value as needed - may need to be readjusted if
% the document is modified later
%\IEEEtriggeratref{8}
% The "triggered" command can be changed if desired:
%\IEEEtriggercmd{\enlargethispage{-5in}}

% references section

% can use a bibliography generated by BibTeX as a .bbl file
% BibTeX documentation can be easily obtained at:
% http://mirror.ctan.org/biblio/bibtex/contrib/doc/
% The IEEEtran BibTeX style support page is at:
% http://www.michaelshell.org/tex/ieeetran/bibtex/
%\bibliographystyle{IEEEtran}
% argument is your BibTeX string definitions and bibliography database(s)
%\bibliography{IEEEabrv,../bib/paper}
%
% <OR> manually copy in the resultant .bbl file
% set second argument of \begin to the number of references
% (used to reserve space for the reference number labels box)
\ifCLASSOPTIONcaptionsoff
  \newpage
\fi

\bibliographystyle{IEEEtran}
%\bibliography{bibtex/bib/IEEEexample}
% Generated by IEEEtran.bst, version: 1.14 (2015/08/26)

\begin{IEEEbiography}[{\includegraphics[width=1in,height=1.25in,clip,keepaspectratio]{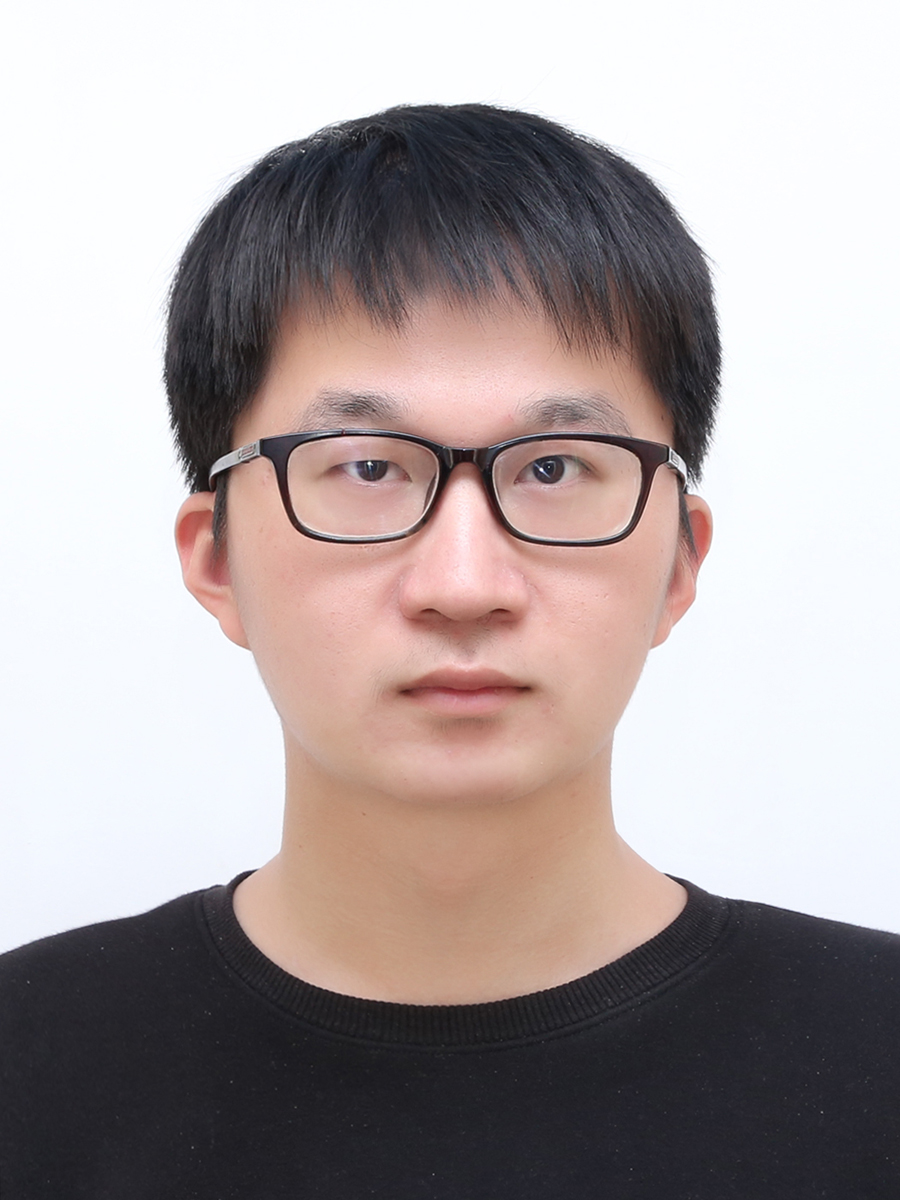}}]{Xiaojie Li} is currently pursuing his Ph.D. at Southeast University. He obtained both his B.Eng. and M.Eng. degrees from the Nanjing University of Aeronautics and Astronautics (NUAA), China, in 2022 and 2025. He was also a visiting student at the Shenzhen Research Institute of Big Data, The Chinese University of Hong Kong-Shenzhen, Guangdong, China. His research interests include spectrum sensing, network optimization, and XL-MIMO technology.
\end{IEEEbiography}

\begin{IEEEbiography}[{\includegraphics[width=1in,height=1.25in,clip,keepaspectratio]{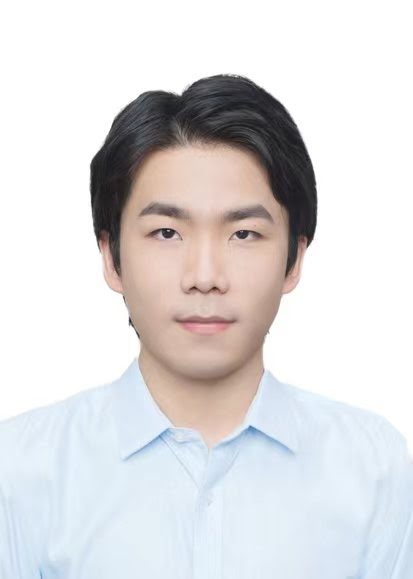}}]{Zhijie Cai} received the B. Sc. degree in Mathematics and Applied Mathematics from the Sun Yat-sen University, Guangdong, China, in 2022. He is currently pursuing his doctoral degree with the School of Science and Engineering, Chinese University of Hong Kong, Shenzhen. His research interests include edge intelligence and distributed machine learning.
\end{IEEEbiography}

\begin{IEEEbiography}[{\includegraphics[width=1in,height=1.25in,clip,keepaspectratio]{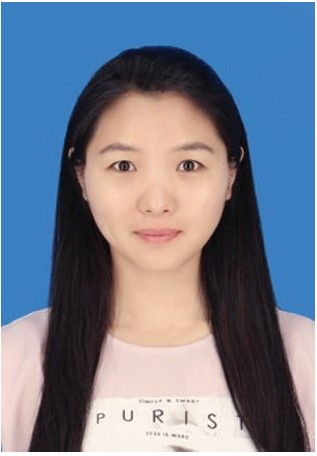}}]{Nan Qi}
(Senior Member, IEEE) received the B.Sc. and Ph.D. degrees in communications engineering from Northwestern Polytechnical University (NPU), China, in 2011 and 2017, respectively. She is currently an associate professor in the Department of Electronic Engineering, Nanjing University of Aeronautics and Astronautics, China. Her research interests include UAV-assisted communications, optimization of wireless communications, opportunistic spectrum access, learning theory, and game theory.
\end{IEEEbiography}

\begin{IEEEbiography}[{\includegraphics[width=1in,height=1.25in,clip,keepaspectratio]{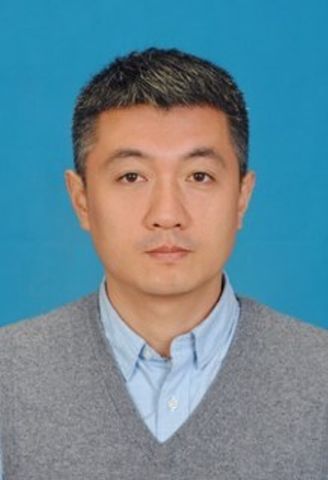}}]{Chao Dong} (Senior Member, IEEE) received the Ph.D. degree in communication engineering from PLA University of Science and Technology, Nanjing, China, in 2007. From 2008 to 2011, he was a
 Postdoctoral Fellow with the Department of Computer Science and Technology, Nanjing University, Nanjing. From 2011 to 2017, he was an Associate Professor with the Institute of Communications Engineering, PLA University of Science and Technol
ogy. He is currently a Full Professor with the College of Electronic and Information Engineering, Nanjing University of Aeronautics and Astronautics, Nanjing. His research interests include D2D communications, UAV swarm networking, and anti-jamming network protocols—e-multimedia communications and wireless communications.
\end{IEEEbiography}

\begin{IEEEbiography}[{\includegraphics[width=1in,height=1.25in,clip,keepaspectratio]{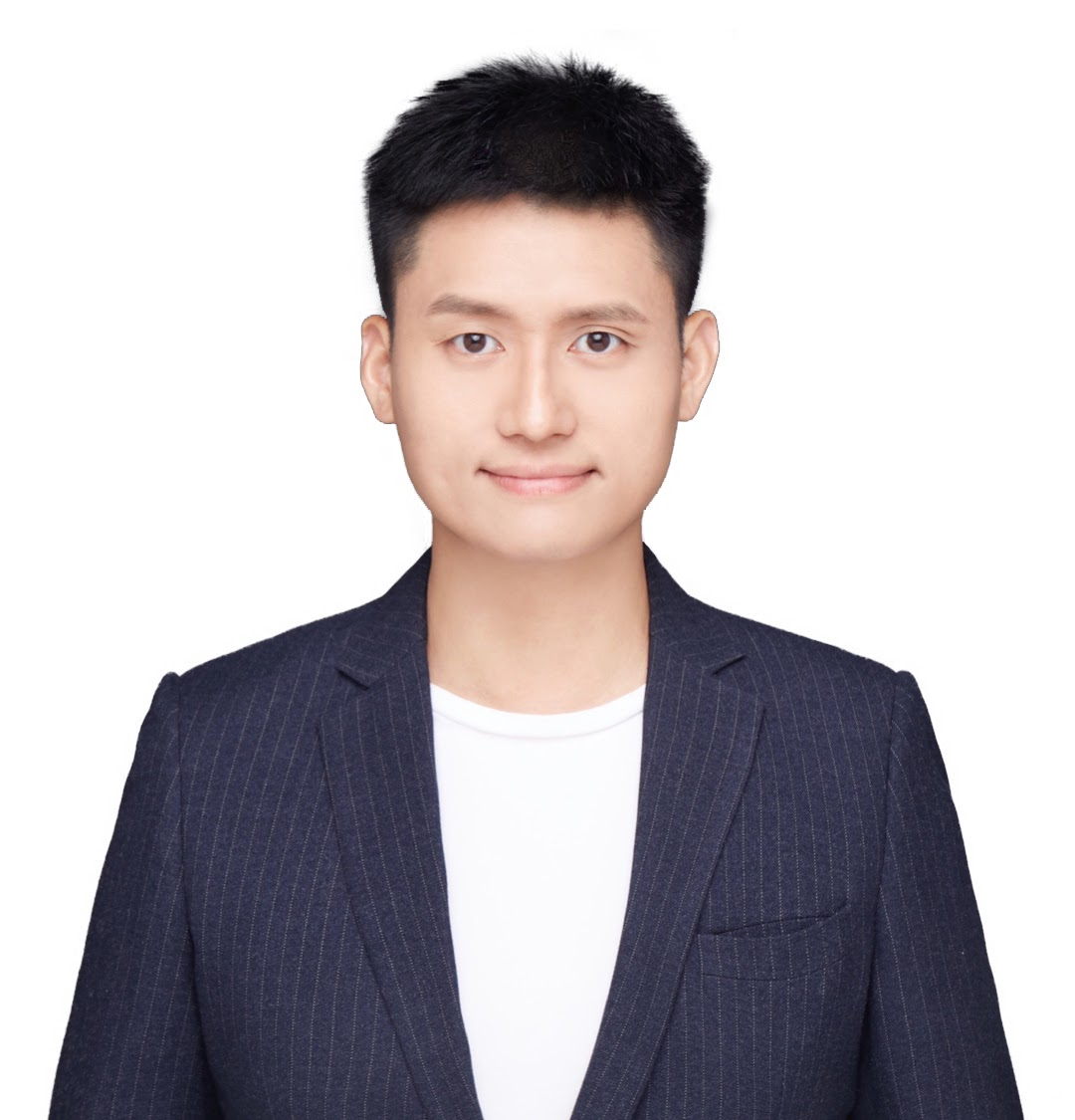}}]{Guangxu Zhu}
(Member, IEEE) received the Ph.D. degree in electrical and electronic engineering from The University of Hong Kong in 2019. Currently he is a senior research scientist and deputy director of network system optimization center at the Shenzhen research institute of big data, and an adjunct associate professor with the Chinese University of Hong Kong, Shenzhen. His recent research interests include edge intelligence, semantic communications, and integrated sensing and communication. He is a recipient of the 2023 IEEE ComSoc Asia-Pacific Best Young Researcher Award and Outstanding Paper Award, the World's Top 2\% Scientists by Stanford University, the "AI 2000 Most Influential Scholar Award Honorable Mention", the Young Scientist Award from UCOM 2023, the Best Paper Award from WCSP 2013 and IEEE  JSnC 2024. He serves as associate editors at top-tier journals in IEEE, including IEEE TMC, TWC and WCL. He is the vice co-chair of the IEEE ComSoc Asia-Pacific Board Young Professionals Committee.
\end{IEEEbiography}

\begin{IEEEbiography}[{\includegraphics[width=1in,height=1.25in,clip,keepaspectratio]{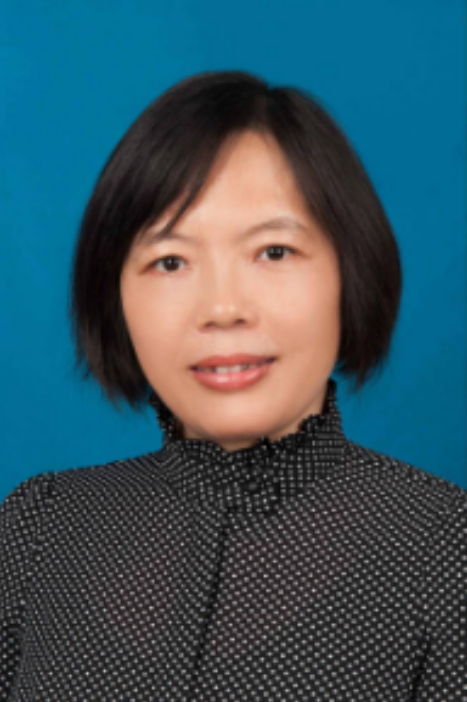}}]{Haixia Ma} received her Master's degree in Radio Physics from Shanghai University, China, in 2002, and her Doctoral degree from the Shanghai Institute of Optics and Fine Mechanics, Chinese Academy of Sciences, China, in 2005. She is currently an Associate Professor at the Department of Physics, Nanjing University of Aeronautics and Astronautics, China. Her research interests include metasurfaces and optical system design.
\end{IEEEbiography}

\begin{IEEEbiography}[{\includegraphics[width=1in,height=1.25in,clip,keepaspectratio]{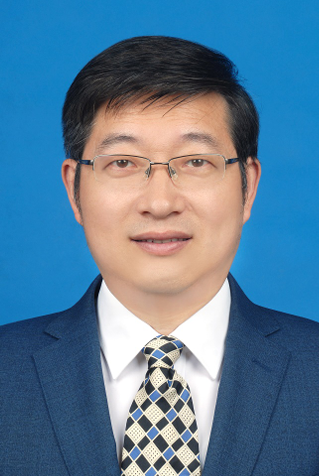}}]{Qihui Wu}  (Fellow, IEEE) received the B.Sc. degree in communications engineering and the M.Sc. and Ph.D. degrees in communications and information systems from the Institute of Communications Engineering, Nanjing, China, in 1994, 1997, and 2000, respectively. From 2003 to 2005, he was a Postdoctoral Research Associate with Southeast University, Nanjing. From 2005 to 2007, he was an Associate Professor
 with the Institute of Communications Engineering, PLA University of Science and Technology, Nanjing, where he is currently a Full Professor. From March 2011 to September 2011, he was an advanced Visiting Scholar with the Stevens Institute of Technology, Hoboken, NJ, USA. Since 2016, he has been with Nanjing University of Aeronautics and Astronautics, Nanjing, and appointed as a Distinguished Professor. His research interests include wireless communications
 and statistical signal processing, with emphasis on system design of software-defined radio, cognitive radio, and smart radio.
\end{IEEEbiography}

\begin{IEEEbiography}[{\includegraphics[width=1in,height=1.25in,clip,keepaspectratio]{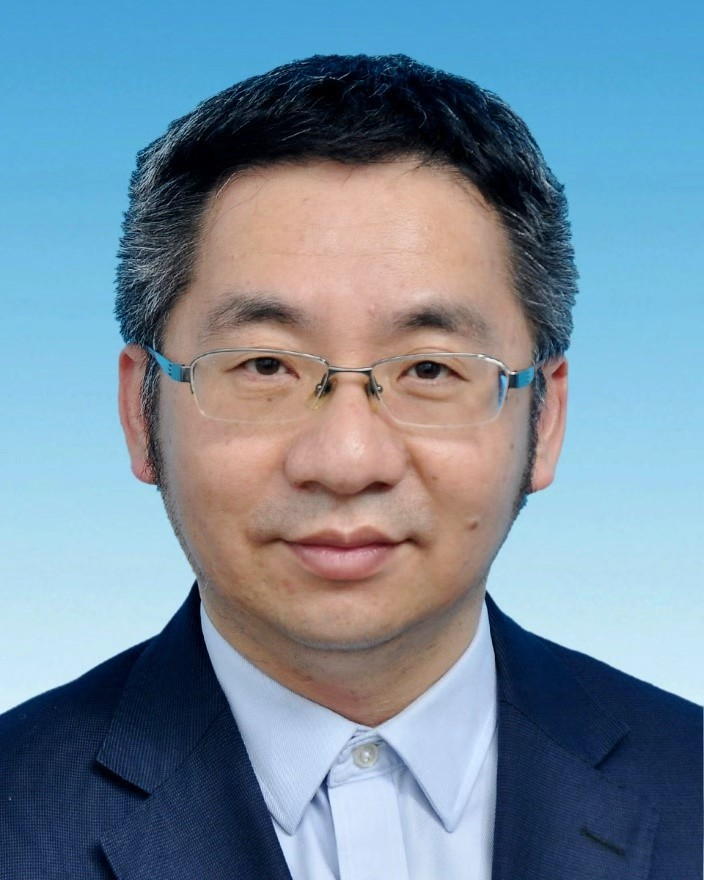}}]{Shi Jin} (Fellow, IEEE) received the B.S. degree in communications engineering from Guilin University of Electronic Technology, Guilin, China, in 1996, the M.S. degree from Nanjing University of Posts and Telecommunications, Nanjing, China, in 2003, and the Ph.D. degree in information and communications engineering from the Southeast University, Nanjing, in 2007. From June 2007 to October 2009, he was a Research Fellow with the Adastral Park Research Campus, University College London, London, U.K. He is currently with the faculty of the National Mobile Communications Research Laboratory, Southeast University. His research interests include wireless communications, random matrix theory, and information theory. He is serving as an Area Editor for the Transactions on Communications and IET Electronics Letters. He was an Associate Editor for the IEEE Transactions on Wireless Communications, and IEEE Communications Letters, and IET Communications. Dr. Jin and his coauthors have been awarded the 2011 IEEE Communications Society Stephen O. Rice Prize Paper Award in the field of communication theory, the IEEE Vehicular Technology Society 2023 Jack Neubauer Memorial Award, a 2022 Best Paper Award and a 2010 Young Author Best Paper Award by the IEEE Signal Processing Society.
\end{IEEEbiography}

\vspace{11pt}

\vfill

% \begin{thebibliography}{1}

% \bibitem{IEEEhowto:kopka}
% H.~Kopka and P.~W. Daly, \emph{A Guide to \LaTeX}, 3rd~ed.\hskip 1em plus
%   0.5em minus 0.4em\relax Harlow, England: Addison-Wesley, 1999.

% \end{thebibliography}

% biography section
% 
% If you have an EPS/PDF photo (graphicx package needed) extra braces are
% needed around the contents of the optional argument to biography to prevent
% the LaTeX parser from getting confused when it sees the complicated
% \includegraphics command within an optional argument. (You could create
% your own custom macro containing the \includegraphics command to make things
% simpler here.)
%\begin{IEEEbiography}[{\includegraphics[width=1in,height=1.25in,clip,keepaspectratio]{mshell}}]{Michael Shell}
% or if you just want to reserve a space for a photo:

% \begin{IEEEbiography}{Michael Shell}
% Biography text here.
% \end{IEEEbiography}

% if you will not have a photo at all:
% \begin{IEEEbiographynophoto}{John Doe}
% Biography text here.
% \end{IEEEbiographynophoto}

% \begin{IEEEbiographynophoto}{Jane Doe}
% Biography text here.
% \end{IEEEbiographynophoto}

% You can push biographies down or up by placing
% a \vfill before or after them. The appropriate
% use of \vfill depends on what kind of text is
% on the last page and whether or not the columns
% are being equalized.

%\vfill

% Can be used to pull up biographies so that the bottom of the last one
% is flush with the other column.
%\enlargethispage{-5in}

% that's all folks
\end{document}